\documentclass[a4paper,onecolumn,accepted=2024-02-13]{quantumarticle}

\pdfoutput=1

\usepackage{amsmath} 
\usepackage{amsthm}        
\usepackage{amssymb}       
\usepackage{braket}       
\usepackage{amsfonts}
\usepackage[all]{xy}
\usepackage{xspace}
\usepackage{calc}
\usepackage{comment}
\usepackage{pgfplots}
\usepgflibrary{fpu}
\usepackage{mathtools}
\usepackage{tabularx}
\usepackage{ifthen}             
\usepackage{graphicx}
\usepackage{docmute}
\usepackage{array}
\usepackage{subfloat} 
\usepackage{textcomp}
\usepackage{bbm}            
\usepackage{etoolbox}  
\usepackage{verbatim}
\usepackage{stmaryrd}
\usepackage{cancel}
\usepackage{xcolor}
\colorlet{Black}{black}
\usepackage{tensor}
\usepackage{enumerate}  
\usepackage{thm-restate}
\usepackage[numbers,sort]{natbib}
\usepackage{enumitem}
\usepackage[export]{adjustbox}

\newenvironment{tz}[1][]{%
                                \begin{tikzpicture}[baseline={([yshift=-.8ex]current bounding                        box.center)},#1] %
                                }{%
                        \end{tikzpicture} %
                        }

\DeclareRobustCommand{\SkipTocEntry}[5]{}

\usepackage{tikz}
\usetikzlibrary{matrix}
\usetikzlibrary{decorations.pathreplacing}
\usetikzlibrary{decorations.markings}
\usetikzlibrary{arrows}
\usetikzlibrary{calc}
\usetikzlibrary{shapes.misc}
\usetikzlibrary{fit}
\usepgflibrary{decorations.pathmorphing}
\usepgflibrary{shapes.geometric}

\pgfdeclarelayer{edgelayer}
\pgfdeclarelayer{nodelayer}
\pgfdeclarelayer{foreground}
\pgfsetlayers{edgelayer,nodelayer,main,foreground}

\tikzstyle{none}=[inner sep=0pt]

\tikzstyle{rn}=[circle,fill=Red,draw=Black,line width=0.8 pt]
\tikzstyle{gn}=[circle,fill=Lime,draw=Black,line width=0.8 pt]
\tikzstyle{bl}=[circle,fill=Blue,draw=Black,line width=0.8 pt]

\tikzstyle{simple}=[-,draw=Black,thick]
\tikzstyle{arrow}=[-,draw=Black,postaction={decorate},decoration={markings,mark=at position .5 with {\arrow{>}}},thick]
\tikzstyle{tick}=[-,draw=Black,postaction={decorate},decoration={markings,mark=at position .5 with {\draw (0,-0.1) -- (0,0.1);}},line width=2.000]
\makeatother\def\thickness{0.7pt}
\tikzstyle{dot}=[circle, draw=black, fill=black, inner sep=.5ex, line width=\thickness, node on layer=foreground]

\makeatletter
\pgfkeys{%
  /tikz/on layer/.code={
    \pgfonlayer{#1}\begingroup
    \aftergroup\endpgfonlayer
    \aftergroup\endgroup
  },
  /tikz/node on layer/.code={
     \gdef\node@@on@layer{%
      \setbox\tikz@tempbox=\hbox\bgroup\pgfonlayer{#1}\unhbox\tikz@tempbox\endpgfonlayer\egroup}
    \aftergroup\node@on@layer
  },
  /tikz/end node on layer/.code={
    \endpgfonlayer\endgroup\endgroup
  }
}

\def\node@on@layer{\aftergroup\node@@on@layer}

\makeatletter
\def\calign@preamble{%
   &\hfil\strut@
    \setboxz@h{\@lign$\m@th\displaystyle{##}$}%
    \ifmeasuring@\savefieldlength@\fi
    \set@field
    \hfil
    \tabskip\alignsep@
}
\let\cmeasure@\measure@
\patchcmd\cmeasure@{\divide\@tempcntb\tw@}{}{}{}
\patchcmd\cmeasure@{\divide\@tempcntb\tw@}{}{}{}
\patchcmd\cmeasure@{\ifodd\maxfields@
  \global\advance\maxfields@\@ne
  \fi}{}{}{}    
\newenvironment{calign}
{%
  \let\align@preamble\calign@preamble
  \let\measure@\cmeasure@
  \align
}
{%
  \endalign
}  
\makeatother
\tikzset{
    master/.style={
        execute at end picture={
            \coordinate (lower right) at (current bounding box.south east);
            \coordinate (upper left) at (current bounding box.north west);
        }
    },
    slave/.style={
        execute at end picture={
            \pgfresetboundingbox
            \path (upper left) rectangle (lower right);
        }
    }
}

\tikzset{blob/.style={draw, circle, fill=white, inner sep=1pt, minimum width=15pt, font=\scriptsize, line width=0.7pt}}
\tikzset{greenregion/.style={fill=green, fill opacity=0.3, draw=none}}
\tikzset{redregion/.style={fill=red, fill opacity=0.3, draw=none}}
\tikzset{blueregion/.style={fill=blue, fill opacity=0.3, draw=none}}
\tikzset{yellowregion/.style={fill=yellow, fill opacity=0.5, draw=none}}
\tikzset{cyanregion/.style={fill=cyan, fill opacity=0.3, draw=none}}
\tikzset{orangeregion/.style={fill=orange, fill opacity=0.6, draw=none}}
\tikzset{solidgreenregion/.style={fill=green!30, fill opacity=1, draw=none}}
\tikzset{solidredregion/.style={fill=red!30, fill opacity=1, draw=none}}
\tikzset{solidblueregion/.style={fill=blue!30, fill opacity=1, draw=none}}
\tikzset{solidyellowregion/.style={fill=yellow!30, fill opacity=1, draw=none}}
\tikzset{string/.style={line width=0.7pt}}
\tikzset{zig/.style={decoration={zigzag,segment length=3, amplitude=0.5}}}
\tikzset{bnd/.style={draw,string}}   
\tikzset{projector/.style={circle, draw, font=\scriptsize, inner sep=-5pt, minimum width=0.35cm, string, fill=white}}
\tikzset{dimension/.style={font=\scriptsize, inner sep=1pt}}

\tikzset{arrow data/.style 2 args={
      decoration={
         markings,
         mark=at position #1 with \arrow{#2}},
         postaction=decorate}
}

\tikzset{along path/.style={every path/.style={}, sloped, allow upside down}}

\def\zxnormal {
                \def \zxscale{0.55}
                \def\zxnodescale{0.8}
                \def\vertexscale{0.7}
                \def\zxshift{0.075cm}
                \def\hadscale{0.8}
                \def\trianglescale{1}
                \def\boxscale{1}
                }

\def\zxgreen{white}

\def\zxwhite{white}

\tikzset{front/.style ={node on layer=foreground}}
\tikzset{zx/.style = {string, scale=\zxscale}}
\tikzset{zxnode/.style n args={1}{blob,scale=\zxnodescale,fill=#1,node on layer=foreground}}
\tikzset{box/.style={draw, rectangle, fill=white, inner sep=1pt, minimum width=10pt,minimum height=10pt, font=\scriptsize, line width=0.7pt,scale=\zxnodescale,node on layer=foreground}}
\tikzset{boxvertex/.style={draw, rectangle, fill=white, line width=0.733pt,scale=0.75*\vertexscale}}
\tikzset{bigbox/.style={draw, rectangle, fill=white,  minimum width=\boxscale *18pt,minimum height=\boxscale*8pt, line width=0.7pt,scale=\zxnodescale}}
\newlength{\unitbox}
\tikzset{widebox/.style ={draw,rectangle, fill=white, line width=0.7pt,scale=0.75*\zxnodescale,minimum height=15pt,inner sep=1pt,  minimum width = \unitbox,   anchor=center }}
\tikzset{wideboxm/.style n args={1}{draw,rectangle, fill=white, line width=0.7pt,scale=0.75*\zxnodescale,minimum height=15pt,inner sep=1pt,  minimum width =2\unitbox+#1\unitbox,   anchor=center }}
\tikzset{triangleup/.style n args={1}{draw, shape=isosceles triangle, isosceles triangle stretches, fill=white, line width=0.7pt,scale=0.75*\zxnodescale,minimum height=15pt,inner sep=1pt,  minimum width = #1*\trianglescale cm +0.15*\trianglescale cm,  shape border rotate=90, anchor=south }}
\tikzset{triangledown/.style n args={1}{draw, shape=isosceles triangle, isosceles triangle stretches, fill=white, line width=0.7pt,scale=0.75*\zxnodescale,minimum height=15pt,inner sep=1pt,  minimum width = #1*\trianglescale cm +0.15*\trianglescale cm,  shape border rotate=-90, anchor=north }}
\tikzset{zxvertex/.style n args={1}{draw,fill=#1,circle,line width=0.7pt,scale=0.75*\vertexscale}}
\tikzset{zxdown/.style={yshift=-\zxshift}}
\tikzset{zxup/.style={yshift=\zxshift}}

\newcommand\mult[3]{ 
\draw[string] (#1.center) to [out=up, in=-135] +(0.5*#2,#3) to [out=-45, in=up] +(0.5*#2,-#3);
\node[zxvertex=\zxgreen,zxdown] at ($(#1)+(0.5*#2,#3)$){};
}

\newcommand\unit[2]{ 
\draw[string] (#1.center) to + (0, -#2);
\node[zxvertex=\zxgreen] at ($(#1) +(0,-#2)$){};
}

\newcommand{\Tr}{\mathrm{Tr}}

\renewcommand{\to}[1][]{\ensuremath{\xrightarrow{#1}}}

\allowdisplaybreaks[1]

\theoremstyle{plain} 

\newtheorem{theorem}{Theorem}[section]
\newtheorem{lemma}[theorem]{Lemma}
\newtheorem{corollary}[theorem]{Corollary}          
\newtheorem{proposition}[theorem]{Proposition}

\theoremstyle{definition} 
\newtheorem{definition}[theorem]{Definition}

\newtheorem{remark}[theorem]{Remark}
\newtheorem{notation}[theorem]{Notation}

\newtheorem{example}[theorem]{Example}
\newtheorem*{example*}{Example}

\theoremstyle{remark}  

\newtheoremstyle{special_statement} 
        {\topskip}
        {\topskip}
        {\addtolength{\leftskip}{2.5em} \itshape }
        {}
        {\bfseries}
        {:}
        {.5em}
        {}
\theoremstyle{special_statement}

\DeclareMathOperator{\Hom}{Hom}
\DeclareMathOperator{\End}{End}

\newcommand{\id}{\mathrm{id}}

\newcommand{\Aut}{\ensuremath{\mathrm{Aut}}}

\newcommand{\Chan}{\mathrm{Chan}}

\newcommand{\Rep}{\mathrm{Rep}}

\newcommand{\CP}{\ensuremath{\mathrm{CP}}}

\newcommand{\Hilb}{\ensuremath{\mathrm{Hilb}}}

\newcommand{\Obj}{\ensuremath{\mathrm{Obj}}}

\newcommand{\Corep}{\ensuremath{\mathrm{Corep}}}

\newcommand{\F}{\ensuremath{\mathrm{SSFA}}}

\makeatletter
\DeclareFontFamily{OMX}{MnSymbolE}{}
\DeclareSymbolFont{MnLargeSymbols}{OMX}{MnSymbolE}{m}{n}
\SetSymbolFont{MnLargeSymbols}{bold}{OMX}{MnSymbolE}{b}{n}
\DeclareFontShape{OMX}{MnSymbolE}{m}{n}{
    <-6>  MnSymbolE5
   <6-7>  MnSymbolE6
   <7-8>  MnSymbolE7
   <8-9>  MnSymbolE8
   <9-10> MnSymbolE9
  <10-12> MnSymbolE10
  <12->   MnSymbolE12
}{}
\DeclareFontShape{OMX}{MnSymbolE}{b}{n}{
    <-6>  MnSymbolE-Bold5
   <6-7>  MnSymbolE-Bold6
   <7-8>  MnSymbolE-Bold7
   <8-9>  MnSymbolE-Bold8
   <9-10> MnSymbolE-Bold9
  <10-12> MnSymbolE-Bold10
  <12->   MnSymbolE-Bold12
}{}

\let\llangle\@undefined
\let\rrangle\@undefined
\DeclareMathDelimiter{\llangle}{\mathopen}%
                     {MnLargeSymbols}{'164}{MnLargeSymbols}{'164}
\DeclareMathDelimiter{\rrangle}{\mathclose}%
                     {MnLargeSymbols}{'171}{MnLargeSymbols}{'171}
\makeatother

\usepackage[colorlinks=true, linkcolor=black, citecolor=black, urlcolor=black, hypertexnames=false
        ]{hyperref}

\newcounter{DRcomment}
\newcounter{DVcomment}
\newcounter{BMcomment}
\newcounter{JVcomment}
\newcommand\ignore[1]{}

\tikzstyle{blackdot}=[circle, draw=black, fill=black, inner sep=.5ex, line width=\thickness, node on layer=foreground]
\tikzstyle{whitedot}=[circle, draw=black, fill=white, inner sep=.5ex, line width=\thickness, node on layer=foreground]

\tikzset{proofdiagram/.style={scale=1}}
\newlength\morphismheight
\newlength\minimummorphismwidth
\newlength\stateheight
\usepackage[utf8]{inputenc}

\title{Entanglement-symmetries of covariant channels}
\author{Dominic Verdon}
\date{School of Mathematics \\ University of Bristol \\[2ex]
16/02/2024}

\usepackage{tikz-cd}
\begin{document}

\normalsize
\zxnormal
\maketitle

\begin{abstract}
Let $G$ and $G'$ be monoidally equivalent compact quantum groups, and let $H$ be a Hopf-Galois object realising a monoidal equivalence between these groups' representation categories. This monoidal equivalence induces an equivalence $\Chan(G) \to \Chan(G')$, where $\Chan(G)$ is the category whose objects are finite-dimensional $C^*$-algebras with an action of $G$ and whose morphisms are covariant channels. We show that, if the Hopf-Galois object $H$ has a finite-dimensional $*$-representation, then channels related by this equivalence can simulate each other using a finite-dimensional entangled resource. We use this result to calculate the entanglement-assisted capacities of certain quantum channels.
\end{abstract}

\maketitle

\section{Introduction}

\subsection{Overview}

\paragraph{Covariant channels.} Symmetry is an important concept in quantum information theory and in mathematics more generally. Using symmetry, a problem which is apparently difficult can be transformed into a problem which is easier to solve.

For quantum channels, the usual notion of symmetry, associated to a group action, is called \emph{covariance}. The definition for compact groups $G$ is as follows. Recall that an \emph{action} of the group $G$ on an finite-dimensional (f.d.) $C^*$-algebra is a continuous homomorphism $\alpha_A: G \to \Aut(A)$, where $\Aut(A)$ is the automorphism group of $A$; we will write $\alpha_{g,A}:= \alpha_A(g) \in \Aut(A)$. We call an f.d. $C^*$-algebra with an action of $G$ an f.d. \emph{$G$-$C^*$-algebra}. A channel $f: A \to B$ between f.d. $G$-$C^*$-algebras (i.e. a completely positive map preserving the canonical $G$-invariant state) is called \emph{covariant} if it intertwines the actions, i.e. $\alpha_{g,B} (f(x)) = f (\alpha_{g,A} (x))$ for all $g \in G, x \in A$. 

When a channel is covariant, it possesses some structure which  allows us to manipulate it. In quantum information theory, entanglement is an important resource for information-processing tasks. In this work we will exhibit a construction which allows covariant channels to be reversibly transformed into each other using quantum entanglement as a resource. 

\paragraph{From groups to Hopf-Galois objects.}

In order to define this construction precisely, we will first motivate it by considering a trivial construction arising from group elements. For any element $g$ of the group $G$ and any f.d. $G$-$C^*$-algebra $A$, we obtain an automorphism $\alpha_{g,A} \in \Aut(A)$. By definition, conjugation by these automorphisms preserves covariant channels; that is, for any covariant $f: A \to B$, we have 
\begin{align}\label{eq:conjugate}
\alpha_{g,B} \circ f \circ \alpha_{g,A}^{-1} = f.
\end{align}
From this perspective, we obtain no interesting transformations of covariant channels by looking at elements of the symmetry group. 

Category theory now provides some useful insight.  For any compact group $G$, consider its category $\Rep(G)$ of finite-dimensional continuous unitary representations. Let $\Hilb$ be the category of finite-dimensional Hilbert spaces and linear maps. We call a faithful unitary $\mathbb{C}$-linear strong monoidal functor $F: \Rep(G) \to \Hilb$ a \emph{fibre functor}. There is a \emph{canonical} fibre functor which takes every representation to its underlying Hilbert space; however, there are in general other, inequivalent fibre functors.

For any element $g$ of the group $G$, and any object $V$ of the category $\Rep(G)$, we obtain a unitary linear map $u_{g,V} \in \End(F(V))$, by acting with the element $g$ on the representation $V$. The data $\{u_{g,V}\}_{V \in \Obj(\Rep(G))}$ defines a \emph{unitary monoidal natural automorphism} of the canonical fibre functor $F$. In fact, all such automorphisms are obtained in this way. We thereby obtain an isomorphism of groups $G \cong \End_u(F)$, where $\End_u(F)$ is the group of unitary monoidal natural automorphisms of the canonical fibre functor.

This observation suggests a generalisation: rather than consider conjugation arising from elements of the symmetry group $G \cong \End_u(F)$ as in~\eqref{eq:conjugate}, we might consider conjugation by unitary transformations between inequivalent fibre functors on $\Rep(G)$. 

For this, we need a more general notion of unitary transformation. This generalisation can be motivated by considering \emph{Hopf-Galois objects}. For every fibre functor $F'$ on $\Rep(G)$, there exists a certain $*$-algebra $H(F')$, the Hopf-Galois object associated to $F'$~\cite{Bichon1999}. Unitary monoidal natural automorphisms of the canonical fibre functor $F$ --- that is, elements of the group $G$ --- correspond precisely to 1-dimensional $*$-representations of $H(F)$.  

We are therefore led to consider finite-dimensional $*$-representations of the Hopf-Galois object $H(F')$ associated to a general fibre functor $F'$. Equivalently~\cite[Thm. 3.14]{Verdon2020}, one can use ideas from 2-category theory and consider \emph{unitary pseudonatural transformations} $F \to F'$.

\paragraph{Entanglement-symmetries of covariant channels.}

We recall that, by Tannaka-Krein-Woronowicz duality~\cite[Thm. 2.3.2]{Neshveyev2013}, every fibre functor $F'$ on $\Rep(G)$ is associated with a compact quantum group $G_{F'}$. Even if the original group $G$ is an ordinary compact group, the group $G_{F'}$ may be a compact quantum group, so we will work in the more general setting of compact quantum groups. An action of a compact quantum group $G$ on an f.d. $C^*$-algebra, and covariance for channels between f.d. $G$-$C^*$-algebras, can be defined analogously to the case of ordinary compact groups~\cite{Wang1998}.

Let $G$ be a compact quantum group, let $F$ be the canonical fibre functor on $\Rep(G)$, let $F'$ be some other fibre functor, and let $G':=G_{F'}$ be the compact quantum group associated to $F'$. In this work we show that this data gives rise to an equivalence of categories $\Chan(G) \to \Chan(G')$, where $\Chan(G)$ is the category of f.d. $G$-$C^*$-algebras and covariant channels. In particular:
\begin{itemize}
\item For every f.d. $G$-$C^*$-algebra $A$, we obtain a corresponding f.d. $G'$-$C^*$-algebra $A'$.
\item For every $G$-covariant channel $f: A \to B$, we obtain a corresponding $G'$-covariant channel $f': A' \to B'$.
\end{itemize}
We saw in the first paragraph of this introduction that, for every element of the group $G$, we obtain an automorphism $\alpha_{g,A} \in \Aut(A)$ for every f.d. $G$-$C^*$-algebra $A$. Now that we have generalised to arbitrary fibre functors $F'$, we should obtain some generalisation of an automorphism. Let $H(F')$ be the Hopf-Galois object associated to $F'$, and let  $(\pi,H_e)$ be a finite-dimensional $*$-representation of $H(F')$. For every f.d. $G$-$C^*$-algebra $A$ we now obtain, instead of an automorphism, a \emph{quantum bijection} $(\alpha_{A},H_e): A \to A'$. The definition of a quantum bijection is as follows~\cite[Lem. 3.6]{Verdon2022}. Let $\Psi: \mathbb{C} \to B(H_e) \otimes B(H_e)$ be the channel initialising a maximally entangled state 
\begin{equation}\label{eq:entstatepuredef}
\ket{\psi}:= \frac{1}{\sqrt{\dim(H_e)}}\sum_{i=1}^{\dim(H_e)}\ket{i}\otimes \ket{i} \in H_e \otimes H_e.
\end{equation}
A quantum bijection $(\alpha_A,H_e):A \to A'$ is  precisely a pair of channels
\begin{align*}
u_{A}: A \otimes B(H_e) \to A' && v_{A}: A' \otimes B(H_e) \to A
\end{align*}
satisfying the following equations (the diagrams are read from bottom to top):
\begin{align*}
\includegraphics[scale=1,valign=c]{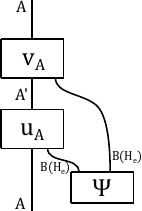}
~~=~~
\includegraphics[scale=1,valign=c]{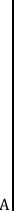}
&&
\includegraphics[scale=1,valign=c]{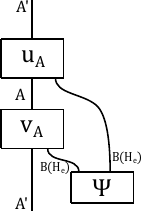}
~~=~~
\includegraphics[scale=1,valign=c]{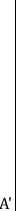}
\\
v_A \circ (u_A \otimes \id_{B(H_e)}) \circ (\id_A \otimes \Psi)~~~~~~ \id_A
&&
u_A \circ (v_A \otimes \id_{B(H_e)}) \circ (\id_{A'} \otimes \Psi)~~~~~~ \id_{A'}
\end{align*}
We can therefore generalise the conjugation action in~\eqref{eq:conjugate} by conjugating, not by automorphisms, but by quantum bijections. In contrast to~\eqref{eq:conjugate}, the action can be non-trivial. In fact, it is non-trivial precisely when the equivalence is nontrivial, since, for every $G$-covariant channel $f: A \to B$, the following equations are obeyed:
\begin{calign}\label{eq:sourcechanentequiv}
\includegraphics[scale=1,valign=c]{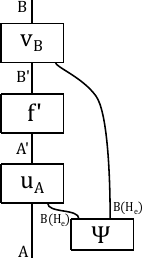}
~~=~~
\includegraphics[scale=1,valign=c]{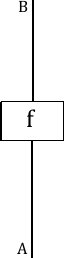}
&&
\includegraphics[scale=1,valign=c]{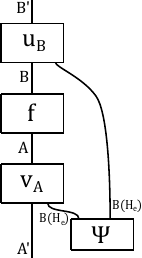}~~=~~
\includegraphics[scale=1,valign=c]{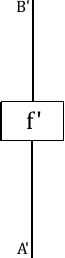}
\end{calign}
We see that the quantum bijections arising from the $*$-representation $(\pi,H_e)$ of $H(F')$ realise the equivalence $\Chan(G) \to \Chan(G')$ specified by the fibre functor $F'$.

We remark that all the maps in~\eqref{eq:sourcechanentequiv} are channels, so the transformation $\Chan(G) \to \Chan(G')$ is physical. We call these transformations arising from finite-dimensional $*$-representations of Hopf-Galois objects for $G$ \emph{entanglement-symmetries of covariant channels}. They are symmetries, not of a single covariant channel alone, but of the whole category $\Chan(G)$.

\paragraph{Entanglement-assisted channel coding.} A natural application of these entanglement-symmetries is to  channel coding. The basic setup is as follows. Alice and Bob can communicate through a channel $N: A \to B$, where $A$ is some input system belonging to Alice, and $B$ is some output system belonging to Bob. They want to communicate through a channel $T: X \to Y$, where $X$ is some input system belonging to Alice and $Y$ is some output system belonging to Bob. A \emph{channel coding scheme for $T$ from $N$} is defined by an encoding channel $E: X \to A$, performed by Alice, and a decoding channel $D: B \to Y$, performed by Bob, such that $D \circ N \circ E = T$:
\begin{calign}\nonumber
&\includegraphics[scale=0.4,valign=c]{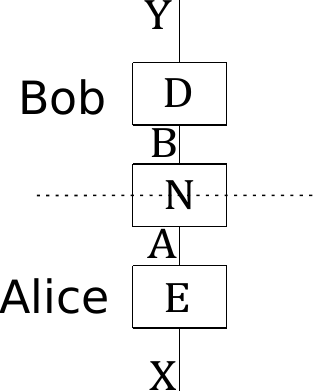}
~~&=~~~
&\includegraphics[scale=0.4,valign=c]{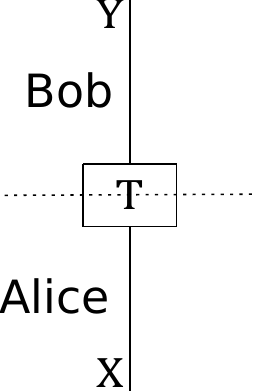}
\\\nonumber
&D \circ N \circ E
~~&~~
&T
\end{calign}
The existence of a channel coding scheme for $T$ from $N$ means that Alice and Bob can communicate through the channel $T$ even through they only share the channel $N$.
\ignore{ More loosely, we can say that the channel $N$ is transformed into the channel $T$ by the encoding and decoding channels $(E,D)$.}

We now introduce entanglement-assistance. Let $H_e$ be a Hilbert space, and let $\Psi: \mathbb{C} \to B(H_e) \otimes B(H_e)$ be the channel initialising the maximally entangled state defined in~\eqref{eq:entstatepuredef}. Operationally, we assume that this state has been initialised beforehand, and that Alice is in possession of one of the entangled systems and Bob the other. To perform her encoding map, Alice may use her half of the maximally entangled state, giving the type $E: X \otimes B(H_e) \to A$ for her encoding operation; to perform his decoding map, Bob may use his half of the maximally entangled state, giving the type $D: B \otimes B(H_e) \to Y$ for his decoding operation. We say that $(E,D,H_e)$ is an \emph{entanglement-assisted channel coding scheme} for $T$ from $N$ if the following equation is obeyed:
\begin{calign}\nonumber
&\includegraphics[scale=1,valign=c]{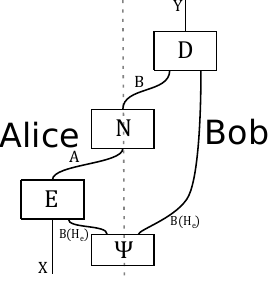}
~~&=~~~
&\includegraphics[scale=1,valign=c]{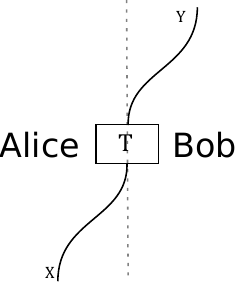}
\\\label{eq:entassistsourcechan}
& D \circ (N \otimes \id_{B(H_e)}) \circ (E \otimes \id_{B(H_e)}) \circ (\id_X \otimes \Psi)
~~&~~
& T
\end{calign}
The existence of an entanglement-assisted channel coding scheme $(E,D,H_e)$ for $T$ from $N$ means that Alice and Bob can communicate through the channel $T$ even though they only share the channel $N$ and and a maximally entangled state of two quantum systems of dimension $\dim(H_e)$. 

It is not hard to show that the relation 
\begin{quote} $(N_1: A_1 \to B_1) \geq (N_2: A_2 \to B_2)$ iff there exists an entanglement-assisted channel coding scheme for $N_2$ from $N_1$
\end{quote} 
defines a partial order on any set of channels. By~\eqref{eq:sourcechanentequiv}, channels related by an entanglement-symmetry of covariant channels are equal under this partial order; that is, they can simulate each other using an entangled resource. 

Knowing that two channels are equivalent in the entanglement-assisted regime is useful. As a first application, in Section~\ref{sec:application} we use entanglement-symmetries of covariant channels to compute the entanglement-assisted capacities of certain quantum channels. This computation is based on the fact that the entanglement-assisted capacities of classical channels are equal to their ordinary capacities. By taking a covariant classical channel and transforming it to a quantum channel using an entanglement-symmetry, we immediately obtain the entanglement-assisted capacity of the resulting quantum channel whenever the capacity of the classical channel is known (since it is identical to that of the classical channel).

\subsection{Structure of the paper}  

In Section~\ref{sec:background} we recall background material on the diagrammatic calculus for rigid $C^*$-tensor categories, unitary pseudonatural transformations between fibre functors (which are precisely finite-dimensional $*$-representations of Hopf-Galois objects), and Tannaka-Krein-Woronowicz duality.

In Section~\ref{sec:algapproach} we set up the Frobenius-algebraic framework for f.d. $G$-$C^*$-algebras that leads to our results. We recall from~\cite{Neshveyev2018} that f.d. $G$-$C^*$-algebras equipped with their canonical faithful $G$-invariant state correspond precisely to \emph{separable standard Frobenius algebras} ($\F$s) in the category $\Rep(G)$. Covariant channels are counit-preserving morphisms between these Frobenius algebras obeying an additional \emph{CP condition}. We put these results in the framework of Tannaka-Krein-Woronowicz duality; given a rigid $C^*$-tensor category $\mathcal{T}$ and a fibre functor $F$ on $\mathcal{T}$ associated to the compact quantum group $G$, one obtains an equivalence between the category of $\F$s in $\mathcal{T}$ and the category $\Chan(G)$. In particular, it follows that the equivalence $\Rep(G) \to \Rep(G')$ corresponding to a fibre functor $F'$ on $\Rep(G)$ with associated compact quantum group $G'$ induces an equivalence $\Chan(G) \to \Chan(G')$.

In Section~\ref{sec:thm} we prove the results about entanglement-symmetries of covariant channels we have outlined in this introduction. 

In Section~\ref{sec:ex} we compute concrete examples of covariant channels related by an entanglement-symmetry, in the case where $\Rep(G)$ is the category $\Hilb(G)$ of $G$-graded Hilbert spaces. We use these examples to compute the entanglement-assisted capacities of some quantum channels, as discussed earlier in this introduction. 

In Appendix~\ref{sec:app} we prove some facts about existence of invariant traces on f.d. $G$-$C^*$-algebras, extending results on $G$-invariant functionals in~\cite{Neshveyev2018}.

\subsection{Related work}

\paragraph{Quantum graph isomorphisms.} In~\cite[\S{}5.1]{Musto2018} the notion of a \emph{quantum graph} $(A,\Gamma)$ was introduced. This is defined by an \emph{adjacency matrix} $\Gamma$ on a f.d. $C^*$-algebra $A$; the adjacency matrix is a certain CP map $\Gamma: A \to A$. A finite-dimensional quantum isomorphism of quantum graphs~\cite[\S{}5.3]{Musto2018} is a quantum bijection intertwining the adjacency matrices. The construction of entanglement-symmetries detailed above holds for general CP maps as well as channels, and it follows from~\cite{Verdon2020,Musto2019} that the adjacency matrices of two quantum isomorphic quantum graphs are related by an entanglement-symmetry of covariant CP maps. (The compact quantum groups associated to the entanglement-symmetry are the quantum automorphism groups of the quantum graphs; c.f.~\cite{Brannan2019}.)

\paragraph{Commuting operator strategies for fully quantum nonlocal games.} One question arising from this work is: can we realise equivalences $\Chan(G) \to \Chan(G')$ corresponding to fibre functors whose associated Hopf-Galois object does not have a finite-dimensional $*$-representation? One would expect the corresponding notion of entanglement-symmetry to arise from infinite-dimensional $*$-representations of Hopf-Galois objects. 

There is a notion of quantum isomorphism of quantum graphs making use of an infinite-dimensional entangled resource, defined in terms of a commuting operator strategy for a nonlocal game satisfying a certain symmetry condition~\cite[\S{}5.4]{Atserias2019}. The connection between these commuting operator quantum isomorphisms and tracial states of Hopf-Galois objects was already made in~\cite[\S{}4.1]{Brannan2019}. In~\cite{Todorov2020,Brannan2020a} a notion of fully quantum nonlocal game was introduced, which can accept inputs and outputs in an arbitrary f.d. $C^*$-algebra. We suspect that entanglement-symmetries making use of an infinite-dimensional entangled resource will be closely connected with commuting operator strategies for fully quantum nonlocal games.

\paragraph{Categorical quantum mechanics.} This work makes extensive use of results obtained in the program of categorical quantum mechanics initiated in~\cite{Abramsky2004}, particularly those relating to the Frobenius-algebraic formulation of finite-dimensional $C^*$-algebras and completely positive maps between them~\cite{Vicary2011,Coecke2016,Heunen2019}. We have also made use of the more general Frobenius-algebraic characterisation of finite-dimensional $C^*$-algebras with a compact quantum group action in~\cite{Neshveyev2018}.  

\subsection{Acknowledgements}

We thank David Reutter and Jamie Vicary for useful discussions. We are grateful to Makoto Yamashita for comments and suggestions regarding invariant traces on f.d. $G$-$C^*$-algebras which led to the inclusion of Appendix~\ref{sec:app}. This work has been funded by the European Research Council (ERC) under the European Union's Horizon 2020 research and innovation programme (grant agreement No. 817581). The work has also been funded by EPSRC (grant number EP/R513179/1).

\ignore{
\section{Notation and conventions}
\label{sec:notation}
All Frobenius algebras considered in this work are dagger, so we omit the qualifier. When referring to Frobenius algebras in some rigid $C^*$-tensor category we write a square-bracketed tuple $[A,m,u]$ where $A$ is the carrier object of the algebra and $m: A \otimes A \to A$ and $u: \mathbbm{1} \to A$ are the multiplication and unit morphisms.

Let $[A_1,m_1,u_1]$, $[A_2,m_2,u_2]$ be separable Frobenius algebras in some rigid $C^*$-tensor category. We show that a morphism $f:A_1 \to A_2$ is CP by writing it as $f: [A_1,m_1,u_1] \to{} [A_2,m_2,u_2]$.

When referring to a unitary corepresentation of a compact quantum group algebra $A_G$ we write a tuple $(V,\rho)$, where $V$ is a Hilbert space and $\rho: V \to V \otimes A_G$ is the corepresentation map. For two unitary corepresentations $(V,\rho_V), (W,\rho_W)$ we show that a linear map $f: V \to W$ is an intertwiner by writing it as $f:(V,\rho_V) \to (W,\rho_W)$.

When referring to a finite-dimensional $G$-$C^*$-algebra for a compact quantum group $G$ of Kac type we write a tuple $([A,m,u],\rho)$ where $[A,m,u]$ is a separable symmetric Frobenius algebra in the category of finite-dimensional Hilbert spaces and linear maps and $\rho: A \to A \otimes A_G$ is an $A_G$-coaction.

Let $([A_1,m_1,u_1],\rho_1)$, $([A_2,m_2,u_2],\rho_2)$ be finite-dimensional $G$-$C^*$-algebras. We show that a channel $f:[A_1,m_1,u_1] \to{} [A_2,m_2,u_2]$ is covariant by writing it as $f:([A_1,m_1,u_1],\rho_1) \to{} ([A_2,m_2,u_2],\rho_2)$.
}

\section{Background}\label{sec:background}
\subsection{Diagrammatic calculus for rigid $C^*$-tensor categories}

A rigid $C^*$-tensor category is in particular a pivotal (in fact, a spherical) dagger category~\cite[\S{}7.3]{Selinger2010} in which the Hom-sets are Banach spaces. For the exact definition, see e.g.~\cite[\S{}2.1]{Neshveyev2013}. We will here only highlight some aspects. 

\paragraph{`Tensor'.}
The adjective `tensor' indicates that a rigid $C^*$-tensor category is monoidal. We use the well-known diagrammatic calculus for monoidal categories~\cite{Heunen2019,Selinger2010}. 

We read diagrams from bottom to top. Objects are drawn as wires, while morphisms are drawn as boxes whose type corresponds to their input and output wires. Composition of morphisms is represented by vertical juxtaposition, while monoidal product is represented by horizontal juxtaposition. For example, two morphisms $f:X\to Y $ and $g: Y \to Z$ can be composed as follows:
\begin{align*}
\includegraphics{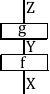}
&&
\includegraphics{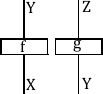} \\
g \circ f: X \to Z 
&&
f \otimes g: X \otimes Y \to Y \otimes Z
\end{align*}
The wire for the monoidal unit $\mathbbm{1}$, and the identity morphism $\id_X$ for any object $X$, are invisible in the diagram. Two diagrams which are planar isotopic represent the same morphism.

\paragraph{`Rigid'.} The adjective `rigid' indicates that a rigid $C^*$-tensor category has duals. Here and throughout when we talk about the duals of a rigid $C^*$-tensor category we mean the standard solutions to the conjugate equations~\cite[Def. 2.2.14]{Neshveyev2013}.
\begin{definition}\label{def:duals}
Let $X$ be an object in a monoidal category. A \emph{dual} $[X^*,\eta,\epsilon]$ for $X$ is:
\begin{itemize}
\item An object $X^*$ (the dual object). To represent dual objects in the graphical calculus, we draw an upward-facing arrow on the $X$-wire and a downward-facing arrow on the $X^*$-wire.
\item Morphisms $\eta_R: \mathbbm{1} \to X^* \otimes X$, $\eta_L: \mathbbm{1} \to X \otimes X^*$, $\epsilon_R: X \otimes X^* \to \mathbbm{1}$, $\epsilon_L: X^* \otimes X \to \mathbbm{1}$, which we draw as cups and caps:
\begin{align*}
\includegraphics[scale=1]{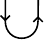}\qquad
&&
\includegraphics[scale=1]{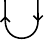}\qquad
&&
\includegraphics[scale=1]{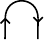}\qquad
&&
\includegraphics[scale=1]{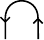}\qquad
\\
\eta_R:\mathbbm{1} \to X^* \otimes X
&&
\eta_L: \mathbbm{1} \to X \otimes X^*
&&
\epsilon_R: X \otimes X^* \to \mathbbm{1}
&&
\epsilon_L: X^* \otimes X \to \mathbbm{1}
\end{align*}
\item These morphisms must satisfy the following \emph{snake equations}: 
\begin{calign}\label{eq:snakes}
\includegraphics[valign=c]{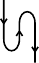}
~~=~~
\includegraphics[valign=c]{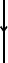}
~~~~~~~~
\includegraphics[valign=c]{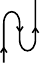}
~~=~~
\includegraphics[valign=c]{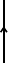}
&&
\includegraphics[valign=c]{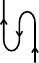}
~~=~~
\includegraphics[valign=c]{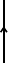}
~~~~~~~~
\includegraphics[valign=c]{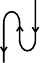}
~~=~~
\includegraphics[valign=c]{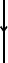}
\end{calign}
\end{itemize}
\end{definition}
\noindent
In a rigid $C^*$-tensor category, the duals allow us to define the dimension of an object and the trace of an endomorphism. Let $X$ be some object and let $f \in \End(X):=\Hom(X,X)$ be a morphism; then the trace of $f$ is defined as follows:
\begin{align}\label{eq:dim}
\includegraphics[valign=c]{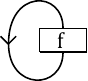}
~~=~~
\includegraphics[valign=c]{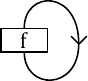}
\end{align}
The fact that the left and right trace are equal here is due to sphericality of a rigid $C^*$-tensor category. The dimension of an object $X$ is defined by $\dim(X):= \Tr(\id_X)$.

\paragraph{`C*'.} The adjective $`C^*{}-$' indicates that a rigid $C^*$-tensor category is a dagger category whose Hom-sets are Banach spaces. Moreover, for any morphism $f: X \to Y$ we have $||f^{\dagger} \circ f|| = ||f||^2$; in particular, for any object $X$ the set $\End(X)$ is a f.d. $C^*$-algebra. 

In the graphical calculus, we give the morphism boxes an offset edge, and represent the dagger of a morphism by reflection in a horizontal axis, preserving the direction of any arrows:
\begin{calign}\label{eq:graphcalcdagger}
\includegraphics[scale=1,valign=c]{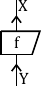}
~~:=~~
\includegraphics[scale=1,valign=c]{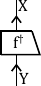}
\end{calign}
\begin{definition}\label{def:isomunitary}
Let $\mathcal{C}$ be a dagger category.  We say that a morphism $\alpha: X \to Y$ is an \emph{isometry} if $\alpha^{\dagger} \circ \alpha = \id_X$. We say that it is a \emph{coisometry} if $ \alpha \circ \alpha^{\dagger}  = \id_Y$. We say that it is \emph{unitary} if it is both an isometry and a coisometry.
\end{definition}
\noindent
In a rigid $C^*$-tensor category, the left cup and cap are the dagger of the right cup and cap (to see that the following equations type check, recall that the arrow directions remain fixed when the diagram is reflected):
\begin{calign}\label{eq:daggerdual}
\includegraphics[scale=1,valign=c]{Figures/svg/2cats/dagdualrcup.pdf}
~~=~~
\left(
\includegraphics[scale=1,valign=c]{Figures/svg/2cats/dagdualrcupflip.pdf}\right)^{\dagger}
&
\includegraphics[scale=1,valign=c]{Figures/svg/2cats/dagduallcap.pdf}
~~=~~
\left(\includegraphics[scale=1,valign=c]{Figures/svg/2cats/dagduallcapflip.pdf}\right)^{\dagger}
\end{calign}

\subsubsection{Example: the category $\Hilb$}
A simple example of a pivotal dagger category is the category $\Hilb$. The objects of the monoidal category $\Hilb$ are finite-dimensional Hilbert spaces, and the morphisms are linear maps between them; composition of morphisms is composition of linear maps. The monoidal product is given on objects by the tensor product of Hilbert spaces, and on morphisms by the tensor product of linear maps; the unit object is the 1-dimensional Hilbert space $\mathbb{C}$.

For any object $H$, its right dual object can be chosen to be $H$ itself; every object is \emph{self-dual}. For this reason we do not generally draw arrows on Hilbert space wires in this work. Indeed, any orthonormal basis $\{\ket{i}\}$ for $H$ defines a cup and cap:
\begin{calign}\label{eq:cupscapsHilb}
\begin{tz}[zx,xscale=-1]
\draw (0,0) to [out=up, in=up, looseness=2.5] (2,0);
\node[dimension, left] at (2.0,0) {$H$};
\node[dimension, right] at (0,0) {$H$};
\node[zxvertex=\zxwhite,zxdown] at (1,1.5){};
\end{tz}
&
\begin{tz}[zx,yscale=-1]
\draw (0,0) to [out=up, in=up, looseness=2.5] (2,0);
\node[dimension, right] at (2.05,0) {$H$};
\node[dimension, left] at (0,0) {$H$};
\node[zxvertex=\zxwhite,zxdown] at (1,1.4){};
\end{tz}\\\nonumber
\ket{i} \otimes \ket{j}\mapsto \delta_{ij}
&
~~1\mapsto \sum_i \ket{i} \otimes \ket{i}
\end{calign}
We draw the cup and cap on $H$ with a white vertex, as in the above diagram. In what follows, we assume that all Hilbert spaces we consider come equipped with some chosen orthonormal basis $\{\ket{i}\}$, defining a self-duality as above. The dagger structure is given by Hermitian adjunction of linear maps. Dimension and trace in the rigid $C^*$-tensor category $\Hilb$ reduce to the usual notions of dimension of a Hilbert space and trace of a linear map.

The category $\Hilb$ is additionally symmetric --- for any $V,W$ there is a \emph{swap map} $\sigma_{V,W}: V \otimes W \to W \otimes V$. We draw this as an intersection:
\begin{align*}
&\includegraphics[scale=1]{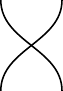} \\
\ket{v} \otimes \ket{w} &\mapsto \ket{w} \otimes \ket{v}
\end{align*}
The four-dimensional calculus allows us to untangle arbitrary diagrams and remove any twists, as exemplified by the following equations:
\begin{calign}
\label{eq:untangling}
\includegraphics[scale=1,valign=c]{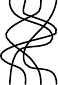}
~~=~~
\includegraphics[scale=1,valign=c]{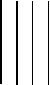}
&
\includegraphics[scale=1,valign=c]{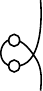}
~~=~~
\includegraphics[scale=1,valign=c]{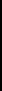}
~~=~~
\includegraphics[scale=1,valign=c]{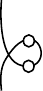}
&
\includegraphics[scale=1,valign=c]{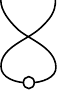}
~~=~~
\includegraphics[scale=1,valign=c]{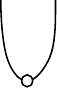}
&
\includegraphics[scale=1,valign=c]{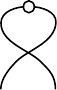}
~~=~~
\includegraphics[scale=1,valign=c]{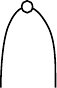}
\end{calign}

\subsection{Fibre functors and unitary pseudonatural transformations}

\subsubsection{Fibre functors}

\begin{definition}\label{def:fibre}
Let $\mathcal{C}$ be a rigid $C^*$-tensor category. We call a faithful unitary $\mathbb{C}$-linear strong monoidal functor $F: \mathcal{C} \to \Hilb$ a \emph{fibre functor}. (Recall that a unitary strong monoidal functor is dagger-preserving, and the multiplicator and unitors associated to the functor are also unitary morphisms in $\Hilb$.)
\end{definition}
\noindent
To depict fibre functors, we use a graphical calculus of \emph{functorial boxes}~\cite{Mellies2006}. In this calculus, we represent the effect of a fibre functor $F: \mathcal{C} \to \Hilb$ by drawing a shaded box around objects and morphisms in $\mathcal{C}$. For example, let $X,Y$ be objects and  $f: X \to Y$ a morphism in $\mathcal{C}$. Then the morphism $F(f): F(X) \to F(Y)$ in $\Hilb$ is represented as:
\begin{calign}\nonumber
\includegraphics[scale=0.8,valign=c]{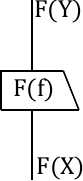}
~~=~~
\includegraphics[scale=0.8,valign=c]{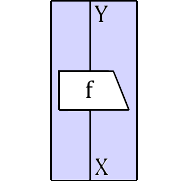}
\end{calign}
We represent the multiplicators $\mu_{X,Y}:F(X) \otimes F(Y) \to F(X \otimes Y)$ and their daggers as follows:
\begin{calign}\nonumber
\includegraphics[scale=.8,valign=c]{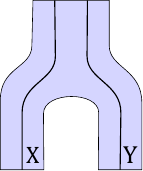}
&
\includegraphics[scale=.7,valign=c]{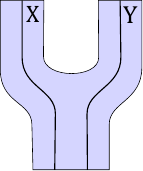}
\\\label{eq:multiplicator}
\mu_{X,Y}: F(X) \otimes  F(Y) \to F(X \otimes Y) & \mu_{X,Y}^{\dagger}: F(X \otimes Y) \to F(X) \otimes F(Y)
\end{calign}
We represent the unitor $\upsilon: \mathbb{C} \to F(\mathbbm{1}_C)$ and its dagger as follows (recall that the monoidal unit is invisible in the diagrammatic calculus):
\begin{calign}\nonumber
\includegraphics[scale=.8,valign=c]{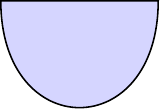}
&
\includegraphics[scale=.8,valign=c]{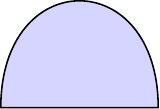} \\\label{eq:unitor}
\upsilon: \mathbb{C} \to F(\mathbbm{1}_C) & \upsilon^{\dagger}: F(\mathbbm{1}_C)  \to \mathbb{C}
\end{calign}

\subsubsection{Unitary pseudonatural transformations}

\begin{definition}\label{def:pntmon}
Let $\mathcal{C}$ be a rigid $C^*$-tensor category and let $F_1,F_2: \mathcal{C} \to \Hilb$ be fibre functors. (We colour the functorial boxes blue and red respectively.) A \emph{unitary pseudonatural transformation} $(\alpha,H): F_1 \to F_2$ is defined by the following data:
\begin{itemize}
\item An Hilbert space $H$ (drawn as a green wire).
\item For every object $X$ of $\mathcal{C}$, a unitary $\alpha_X: F_1(X) \otimes H \to H \otimes F_2(X)$ (drawn as a white vertex):
\begin{calign}\nonumber
\includegraphics[scale=1.2,valign=c]{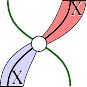}
\end{calign}
\end{itemize}
The unitaries $\{\alpha_X\}$ must satisfy the following conditions:
\begin{itemize}
\item \emph{Naturality.} For every morphism $f:X \to Y$ in $\mathcal{C}$:
\begin{calign}\label{eq:pntmonnat}
\includegraphics[scale=1,valign=c]{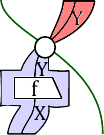}
~~=~~
\includegraphics[scale=1,valign=c]{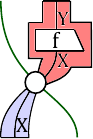}
\end{calign}
\item \emph{Monoidality.} 
\begin{itemize}
\item For every pair of objects  $X,Y$ of $\mathcal{C}$:
\begin{calign}\label{eq:pntmonmon}
\includegraphics[scale=1,valign=c]{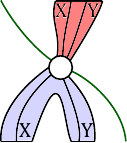}
~~=~~
\includegraphics[scale=1,valign=c]{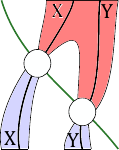}
\end{calign}
\item $\alpha_{\mathbbm{1}}$ is defined as follows:
\begin{calign}\label{eq:pntmonmonunit}
\includegraphics[scale=1,valign=c]{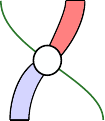}
~~=~~
\includegraphics[scale=1,valign=c]{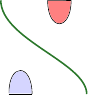}
\end{calign}
\end{itemize}
\end{itemize}
\end{definition}
\begin{remark}\label{rem:monnattransfisupt}
Unitary pseudonatural transformations generalise unitary monoidal natural isomorphisms, which we recover when $H \cong \mathbb{C}$.
\end{remark}
\begin{remark}\label{rem:cleftdimpres}
Since the components of the UPT are unitary, it follows that $\dim(F_1(X)) = \dim(F_2(X))$ for any object $X$ in $\mathcal{C}$.
\end{remark}

\ignore{
\begin{remark}
The diagrammatic calculus shows that pseudonatural transformation is a planar notion. The $H$-wire forms a boundary between two regions of the $\mathcal{D}$-plane, one in the image of $F$ and the other in the image of $G$. By pulling through the $H$-wire, morphisms from $\mathcal{C}$ can move between the two regions~\eqref{eq:pntmonnat}. 
\end{remark}
\noindent
UPTs $(\alpha,H): F_1 \to F_2$ and $(\beta,H'): F_2 \to F_3$ can be composed associatively to obtain a UPT $(\alpha \circ \beta,H \otimes H'): F_1 \to F_3$ whose components $(\alpha \circ \beta)_X$ are as follows (we colour the $H'$-wire orange, and the $F_3$-box brown):
\begin{calign}
\includegraphics[scale=1]{Figures/svg/uptrecap/comptransfs.png}
\end{calign}
}
\begin{definition}[{\cite[Lem. 3.3, Lem. 4.2] {Verdon2020a}}]
Let $(\alpha,H): F_1 \to F_2$ be a UPT. Then the \emph{dual} of $\alpha$ is a UPT $(\alpha^*,H): F_2 \to F_1$ whose components $(\alpha^*)_X$ are defined as follows:
\begin{calign}\label{eq:dualpnt}
\includegraphics[valign=c]{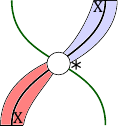}
:=
\includegraphics[scale=1.1,valign=c]{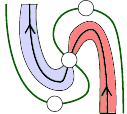}
=
\includegraphics[scale=1,valign=c]{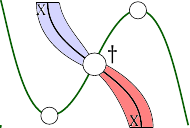}
\end{calign}
\end{definition}
\noindent
It can be shown~\cite[Thm. 4.4, Cor. 5.6]{Verdon2020a} that $\alpha$ and $\alpha^*$ obey the following pull-through equations:
\begin{calign}\nonumber
\includegraphics[valign=c]{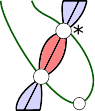}
~~=~~
\includegraphics[valign=c]{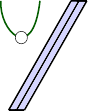}
&
\includegraphics[valign=c]{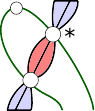}
~~=~~
\includegraphics[valign=c]{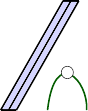}
\\\label{eq:cupcapmodsdualpnt}
\includegraphics[valign=c]{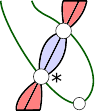}
~~=~~
\includegraphics[valign=c]{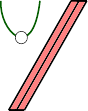}
&
\includegraphics[valign=c]{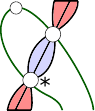}
~~=~~
\includegraphics[valign=c]{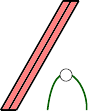}
\end{calign}
\begin{remark}UPTs from the canonical fibre functor on a rigid $C^*$-tensor category, as well as the fibre functors themselves, were classified in terms of Hopf-Galois theory in~\cite[Thm. 3.14]{Verdon2020}.
\end{remark}

\subsection{Compact quantum groups and Tannaka-Krein-Woronowicz reconstruction}

\subsubsection{Compact quantum groups}

A quite general example of a rigid $C^*$-tensor category is the category of representations of a \emph{compact quantum group}. 
\begin{definition}[{\cite[Definition 1.6.1]{Neshveyev2013}}]
A unital $*$-algebra $A$ equipped with a unital $*$-homomorphism $\Delta: A \to A \otimes A$ (the \emph{comultiplication}) is called a \emph{Hopf-$*$-algebra} if $(\Delta \otimes \id_A) \circ \Delta = (\id_A \otimes \Delta) \circ \Delta$ and there exist linear maps $\epsilon: A \to \mathbb{C}$ (the \emph{counit}) and $S: A \to A$ (the \emph{antipode}) such that 
\begin{calign}\nonumber
(\epsilon \otimes \id_A) \circ \Delta = \id_A = (\id_A \otimes \epsilon) \circ \Delta 
&&
m \circ (S \otimes \id_A) \circ \Delta = u \circ \epsilon = m \circ (\id_A \otimes S) \circ \Delta
\end{calign}
where $m: A \otimes A \to A$ is the multiplication and $u: \mathbb{C} \to A$ the unit of the algebra $A$.
\end{definition}
\ignore{
\begin{example}
Let $G$ be a compact group, and let $C(G)$ be the $\mathbb{C}^*$-algebra of continuous complex functions on $G$. For any finite-dimensional $*$-representation $\pi: G \to B(H)$, each matrix entry $\pi_{ij}: G \to \mathbbm{C}$ is an element of $C(G)$. Let $\mathbb{C}[G] \subset C(G)$ be the linear span of matrix coefficients of finite-dimensional representations of $G$. This naturally possesses the structure of a Hopf $*$-algebra (see~\cite{}).
\end{example}
}
\begin{definition}[{\cite[Definition 1.6.5]{Neshveyev2013}}]\label{def:unitarycorep}
A \emph{corepresentation} $(H,\delta)$ of a Hopf $*$-algebra $A$ on a vector space $H$ is defined by a linear map $\delta: H \to H \otimes A$ such that 
\begin{calign}\nonumber
(\delta \otimes \id_A) \circ \delta = (\id_H \otimes \Delta) \circ \delta 
&&
(\id_H \otimes \epsilon) \circ  \delta = \id_H
\end{calign}
Suppose that $H$ is a Hilbert space and let $\{\ket{i}\}$ be an orthonormal basis. We call the elements $\delta_{ij}:=(\bra{i} \otimes \id_A) \circ \delta (\ket{j})$ of $A$ the \emph{matrix coefficients} of $(H,\delta)$.
We then say that the corepresentation $(H,\delta)$ is \emph{unitary} if the matrix coefficients obey the following equation:
$$
\sum_{j} \delta_{ji}^* \delta_{jk} = 1_A = \sum_{j} \delta_{ij} \delta_{kj}^*
$$
For $(H_1,\delta_1), (H_2,\delta_2)$ corepresentations of a Hopf-$*$-algebra $A$, we say that a linear map $f: H_1 \to H_2$ is an \emph{intertwiner} $f:(H_1,\delta_1) \to (H_2,\delta_2)$ if $\delta_2 \circ f = (f \otimes \id_A) \circ \delta_1$.
\end{definition}
\begin{definition}[{c.f. \cite[Theorem 1.6.7]{Neshveyev2013}}]
We say that a Hopf-$*$-algebra is a \emph{compact quantum group algebra} if it is generated as an algebra by matrix coefficients of its finite-dimensional unitary corepresentations. 
\end{definition}
\noindent
The category $\Corep(A)$ whose objects are finite-dimensional unitary corepresentations of a compact quantum group algebra $A$ and whose morphisms are intertwiners is a rigid $C^*$-tensor category~\cite[\S{}1.6]{Neshveyev2013}.
\begin{example}\label{ex:compactgroupcoreps}
Let $G$ be a compact group. A continuous function $f:G \to \mathbb{C}$ is called \emph{representative} if there exists a continuous representation $\pi$ of $G$ on some finite-dimensional complex vector space $H$ and elements $v \in H$, $\phi \in H^*$ such that $f(x) = \phi(\pi(x)v)$ for all $x \in G$. The algebra of representative functions $R(G)$ has the structure of a Hopf-$*$-algebra where the comultiplication, counit and antipode are defined as follows:
\begin{align*}
\Delta(f)(x,y) = f(xy) 
&&
\epsilon(f) = f(e)
&&
(S(f))(x) = f(x^{-1})
\end{align*}
The Hopf-$*$-algebra $R(G)$ is generated by the coefficients of its finite-dimensional unitary corepresentations; in other words, it is a compact quantum group algebra. 

There is an isomorphism of categories between the category of continuous unitary representations of $G$ and the category of unitary corepresentations of $R(G)$ (see~\cite[Example 3.1.5]{Timmerman2008} for the bijection on objects; the extension to morphisms is obvious). Roughly, for any continuous unitary representation $\pi$ of $G$ on $H$ we may identify $\End(H)$ with $M_n(\mathbb{C})$ by some orthonormal basis $\{\ket{i}\}$ and consider $\pi$ as an $R(G)$-valued matrix $(\pi_{ij})$; this defines a unitary corepresentation $\pi^c: H \to H \otimes R(G)$ whose matrix coefficients are $(\bra{i} \otimes \id_{R(G)}) \pi^c \ket{j} = \pi_{ij}$.
\end{example}
\noindent
In the spirit of Example~\ref{ex:compactgroupcoreps}, a compact quantum group algebra $A$ is considered as the algebra of matrix coefficents of continuous unitary representations of some `compact quantum group' $G$, such that $\Rep(G) = \Corep(A)$. We will refer to compact quantum groups $G$, and write $\Rep(G)$, in order to emphasise the similarity with the representation theory of compact groups. However, the compact quantum group algebra is the concrete object, at least as far as this work is concerned. We write $A_G$ for the compact quantum group algebra corresponding to a compact quantum group $G$. 

\begin{definition}\label{def:quantumdimension}
For a f.d. unitary representation $X$ of a compact quantum group $G$ we call the dimension~\eqref{eq:dim} of $X$ as an object in the rigid $C^*$-tensor category $\Rep(G)$ its \emph{quantum dimension} $\dim_q(X)$. We call the dimension of the underlying Hilbert space of $X$ the \emph{ordinary dimension} $\dim_c(X)$. \ignore{ We say that a compact quantum group $G$ is \emph{of Kac type} if $\dim_q(X) = \dim_c(X)$ for every f.d. unitary representation $X$ of $G$. }
\end{definition}

\subsubsection{Tannaka-Krein-Woronowicz duality}
\label{sec:tannaka}

The category $\Rep(G)$ has a \emph{canonical} fibre functor which takes each corepresentation to its underlying Hilbert space and each intertwiner to its underlying linear map.
\begin{theorem}[{\cite[Theorem 2.3.2]{Neshveyev2013}}]\label{thm:tannaka}
Let $\mathcal{C}$ be a rigid $C^*$-tensor category and let $U: \mathcal{C} \to \Hilb$ be a fibre functor.
There exists a compact quantum group $G$ and a unitary monoidal equivalence $E: \mathcal{C} \to \Rep(G)$ such that the following diagram of functors commutes\footnote{In~\cite[Theorem 2.3.2]{Neshveyev2013} it is only stated that the diagram commutes up to unitary monoidal natural isomorphism. In fact for the construction given in~\cite[\S{}2.3]{Neshveyev2013}, it commutes on-the-nose.}, where $F: \Rep(G) \to \Hilb$ is the canonical fibre functor:
\begin{equation}\label{diag:tannaka}
\begin{tikzcd}
\mathcal{C} \arrow[r, "E"] \arrow[rd, "U"] & \Rep(G) \arrow[d,"F"] \\
& \Hilb
\end{tikzcd}
\end{equation}
\end{theorem}
\noindent
In particular, for any object $V$ of $\mathcal{C}$ one can define  a unitary corepresentation $\rho_{U,V}: U(V) \to U(V) \otimes A_G$. These unitary corepresentations $\{\rho_{U,V}\}_{V \in \Obj(\mathcal{C})}$ have the following properties:
\begin{itemize}
\item \emph{Naturality.} For any $f: V \to W$ in $\mathcal{C}$, the morphism $U(f): (U(V),\rho_{U,V}) \to (U(W),\rho_{U,W})$ is an intertwiner, i.e.:
\begin{equation}\label{eq:rhonat}
\rho_{U,W} \circ U(f) = (U(f) \otimes \id_{A_G}) \circ \rho_{U,V}
\end{equation}
\item \emph{Monoidality.} We have the following equations, where the black vertex represents the multiplication and $u_{A_G}$ is the unit of the Hopf-$*$-algebra $A_G$:
\begin{itemize}
\item For any objects $V,W$ in $\mathcal{C}$:
\begin{calign}\label{eq:rhomon}
\includegraphics[scale=.6,valign=c]{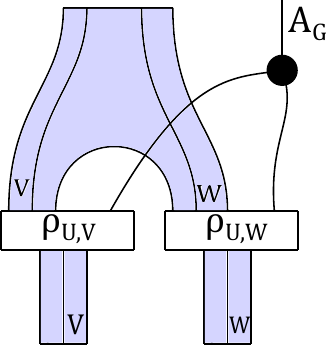}
~~=~~
\includegraphics[scale=.6,valign=c]{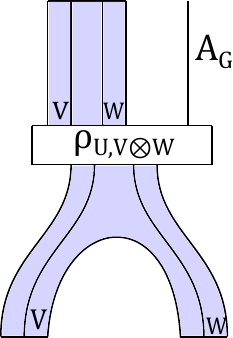}
\end{calign}
\item For the monoidal unit object $\mathbbm{1}$ in $\mathcal{C}$:
\begin{equation}\label{eq:rhomonunit}
\rho_{U,\mathbbm{1}} = \id_{U(\mathbbm{1})} \otimes u_{A_G}
\end{equation}
\ignore{
\begin{calign}
\includegraphics[scale=.6,valign=c]{Figures/svg/tannaka/tannakacomodmon21.png}
~~=~~
\includegraphics[scale=.6,valign=c]{Figures/svg/tannaka/tannakacomodmon22.png}
\end{calign}
}
\end{itemize}
\end{itemize}
The unitary monoidal equivalence $E: \mathcal{C} \to \Rep(G)$ is then defined on objects by $X \mapsto (U(X),\rho_{U,X})$ and on morphisms by $f \mapsto U(f)$. It is clear that the diagram~\eqref{diag:tannaka} commutes on-the-nose. 

\section{Frobenius algebras and f.d. $G$-$C^*$-algebras}\label{sec:algapproach}

In this section we will reformulate finite-dimensional $G$-$C^*$-algebra theory in terms of Frobenius algebras. Essentially everything in this section is well-known~\cite{Kock2003,Vicary2011,Bischoff2015,Coecke2016,Neshveyev2018,Heunen2019}; we have just brought it together in one place. To our knowledge the formulation of f.d. $G$-$C^*$-algebra theory in terms of Tannaka duality in Section~\ref{sec:gcstaralg} is new, but this is a very straightforward consequence of the results in~\cite[\S{}2.3]{Neshveyev2018}.

\subsection{Frobenius algebras}\label{sec:frobrecap}
First, we recall basic definitions regarding Frobenius algebras.  
\begin{definition}\label{def:Frobenius}
Let $\mathcal{C}$ be a rigid $C^*$-tensor category. An \emph{algebra} $[A,m,u]$ in $\mathcal{C}$ is an object $A$ with multiplication and unit morphisms, depicted as follows:%
\begin{calign}\nonumber
\begin{tz}[zx,master]
\coordinate (A) at (0,0);
\draw (0.75,1) to (0.75,2);
\mult{A}{1.5}{1}
\end{tz}
&
\begin{tz}[zx,slave]
\coordinate (A) at (0.75,2);
\unit{A}{1}
\end{tz}
\\[0pt]\nonumber
m:A\otimes A \to A& u: \mathbbm{1} \to A 
\end{calign}\hspace{-0.2cm}
These morphisms satisfy the following associativity and unitality equations:
\begin{calign}\label{eq:assocandunitality}
\includegraphics[scale=1,valign=c]{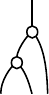}
~~=~~
\includegraphics[scale=1,valign=c]{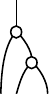}
&
\includegraphics[scale=1,valign=c]{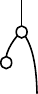}
~~=~~
\includegraphics[scale=1,valign=c]{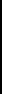}
~~=~~
\includegraphics[scale=1,valign=c]{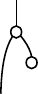}
\end{calign}
Analogously, a \textit{coalgebra} $[A,\delta,\epsilon]$ is an object $A$ with a  comultiplication $\delta: A \to A\otimes A$ and a counit $\epsilon:A\to \mathbb{C}$ obeying the following coassociativity and counitality equations:
\begin{calign}\label{eq:coassocandcounitality}
\includegraphics[scale=1,valign=c]{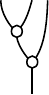}
~~=~~
\includegraphics[scale=1,valign=c]{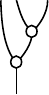}
&
\includegraphics[scale=1,valign=c]{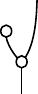}
~~=~~
\includegraphics[scale=1,valign=c]{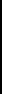}
~~=~~
\includegraphics[scale=1,valign=c]{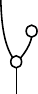}
\end{calign}
The dagger of an algebra $[A,m,u]$ is a coalgebra $[A,m^{\dagger},u^{\dagger}]$.  A algebra $[A,m,u]$ in $\mathcal{C}$ is called \textit{Frobenius} if the algebra and adjoint coalgebra structures are related by the following \emph{Frobenius equation}:
\begin{calign}\label{eq:Frobenius}
\includegraphics[scale=1,valign=c]{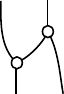}
~~=~~
\includegraphics[scale=1,valign=c]{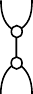}
~~=~~ 
\includegraphics[scale=1,valign=c]{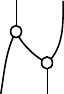}
\end{calign}
\end{definition}
\begin{remark}
Frobenius algebras are canonically self-dual. Indeed, it is easy to see that for any Frobenius algebra the following cup and cap fulfil the snake equations~\eqref{eq:snakes}:
\begin{calign}\label{eq:cupcapfrob}
\includegraphics[scale=1,valign=c]{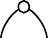}
~~:=~~
\includegraphics[scale=1,valign=c]{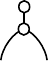}
&
\includegraphics[scale=1,valign=c]{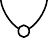}
~~:=~~
\includegraphics[scale=1,valign=c]{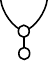}
\end{calign}
We say that the Frobenius algebra is \emph{standard} if, for any $f \in \End(A)$, we have the following equality:
\begin{align*}
\includegraphics[scale=1,valign=c]{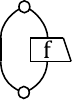}
~~=~~
\includegraphics[scale=1,valign=c]{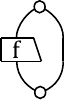}
\end{align*}
\end{remark}
\noindent
A Frobenius algebra is \emph{separable}\footnote{Here we have followed the definition of separability of a Frobenius algebra common in the categorical quantum mechanics literature under the name of \emph{specialness} (e.g.~\cite{Heunen2013,Heunen2019,Coecke2013,Coecke2016,Vicary2011}), namely $mm^{\dagger} = \id$.) 
} if the following equation is satisfied:
\begin{calign}\label{eq:frobspecial}
\includegraphics[scale=1,valign=c]{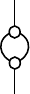}
~~=~~ 
\includegraphics[scale=1,valign=c]{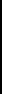}
\end{calign}
If $\mathcal{C}$ is additionally symmetric (for instance, $\mathcal{C}=\Hilb$) a Frobenius algebra is \emph{symmetric} or furthermore \emph{commutative} if the following equations are satisfied, where symmetry is the left equation and commutativity the right:
\begin{calign}\label{eq:frobsymm}
\includegraphics[scale=1,valign=c]{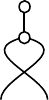}
~~=~~ 
\includegraphics[scale=1,valign=c]{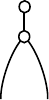}
&
\includegraphics[scale=1,valign=c]{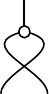}
~~=~~ 
\includegraphics[scale=1,valign=c]{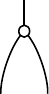}
\end{calign}
\begin{definition}\label{def:starhomcohom}
Let $[A,m_A,u_A],[B,m_B,u_B]$ be Frobenius algebras in a rigid $C^*$-tensor category. We say that a morphism $f: A \to B$ is a \emph{$*$-homomorphism} precisely when it satisfies the following equations:
\begin{calign}\label{eq:homo}
\includegraphics[scale=1,valign=c]{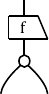}
~~=~~
\includegraphics[scale=1,valign=c]{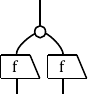}
&
\includegraphics[scale=1,valign=c]{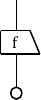}
~~=~~
\includegraphics[scale=1,valign=c]{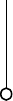}
&
\includegraphics[scale=1,valign=c]{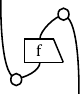}
~~=~~
\includegraphics[scale=1,valign=c]{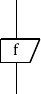}
\end{calign}
We say that a morphism $f: A \to B$ is a \emph{$*$-cohomomorphism} precisely when it satisfies the following equations:
\begin{calign}\label{eq:cohomo}
\includegraphics[scale=1,valign=c]{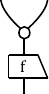}
~~=~~
\includegraphics[scale=1,valign=c]{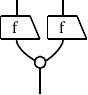}
&
\includegraphics[scale=1,valign=c]{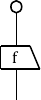}
~~=~~
\includegraphics[scale=1,valign=c]{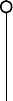}
&
\includegraphics[scale=1,valign=c]{Figures/svg/frobalgs/starhom31right.pdf}
~~=~~
\includegraphics[scale=1,valign=c]{Figures/svg/frobalgs/starhom32.pdf}
\end{calign}
A \emph{unitary $*$-isomorphism} $[A,m_A,u_A] \to{} [B,m_B,u_B]$ is a unitary $*$-homomorphism or $*$-cohomomorphism. (It is easy to check that unitarity of $f$ means that either of~\eqref{eq:homo} or~\eqref{eq:cohomo} implies the other.)

There is a weaker relevant notion of isomorphism: an  \emph{$*$-isomorphism} of Frobenius algebras is an invertible $*$-homomorphism. A $*$-isomorphism of Frobenius algebras can equivalently be defined as an invertible morphism $f: A \to B$ obeying the first two equations of~\eqref{eq:homo} and such that $f^{\dagger}\circ f$ is a left $A$-module map (see~\cite[Lem. 2.3]{Neshveyev2013}).

A $*$-isomorphism of Frobenius algebras is a unitary $*$-isomorphism precisely when it preserves the counit (the second equation of~\eqref{eq:cohomo}).
\end{definition}
\noindent
In a symmetric rigid $C^*$-tensor category such as $\Hilb$, we may define the tensor product of Frobenius algebras. 
\begin{definition}\label{def:tensorprodfrob}
Let $[A,m_A,u_A]$ and $[B,m_B,u_B]$ be Frobenius algebras in a symmetric monoidal dagger category. Then their \emph{tensor product} is a Frobenius algebra $[A \otimes B,m_{A \otimes B},u_{A \otimes B}]$, where $m_{A \otimes B}$ and $u_{A \otimes B}$ are defined as follows:
\begin{calign}
\includegraphics[scale=1,valign=c]{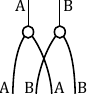}
&&
\includegraphics[scale=1,valign=c]{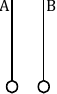}
\end{calign}
It is straightforward to show that the tensor product preserves separability and symmetry.
\end{definition}

\subsection{Finite-dimensional $C^*$-algebras and channels}

We will now recall that Frobenius algebras in the category $\Hilb$ correspond precisely to finite-dimensional $C^*$-algebras equipped with a faithful positive linear functional. For one direction, we observe that any finite-dimensional $C^*$-algebra $A$ equipped with a faithful positive linear functional $\phi: A \to \mathbb{C}$ may be given the structure of a Hilbert space using the inner product $\braket{x|y} = \phi(x^*y)$. We thereby obtain an algebra $[A,m,u]$ in $\Hilb$. In fact, this algebra is Frobenius. 
\begin{lemma}
The algebra $[A,m,u]$ in $\Hilb$ corresponding to a finite-dimensional $C^*$-algebra $A$ with faithful positive linear functional $\phi: A \to \mathbb{C}$ is Frobenius. Moreover, we have $\phi = u^{\dagger}$ and the involution may be expressed as follows:
\begin{calign}\label{eq:frobinvol}
\includegraphics[scale=.6,valign=c]{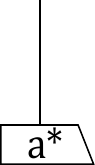}
~~:=~~
\includegraphics[scale=.6,valign=c]{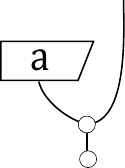}
\end{calign}
\end{lemma}
\begin{proof}
We first observe that the involution takes the form~\eqref{eq:frobinvol}. Indeed, we observe that for any $a,b \in A$ we have $\braket{ab | 1} = \phi((ab)^*) = \phi(b^*a^*)= \braket{b|a^*}$. Depicting this graphically, we obtain the following equation for any $a \in A$:
\begin{calign}\nonumber
\includegraphics[scale=.6,valign=c]{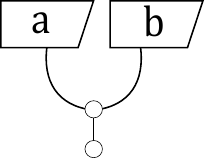}~~=~~\includegraphics[scale=.6]{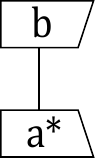}\qquad \qquad \qquad \forall~ b \in A
\end{calign} 
Since the equation holds for all $b \in A$, it follows that the involution takes the form~\eqref{eq:frobinvol}.

We now show that $A$ is a Frobenius algebra. We observe the following implication of the equality $\braket{1|b^*a}=\phi(b^*a)= \braket{b|a}$:
\begin{calign}\label{eq:frobpfsnake}
\includegraphics[scale=.6,valign=c]{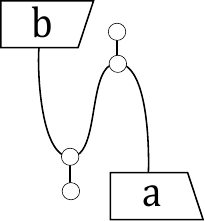}
~~=~~
\includegraphics[scale=.6,valign=c]{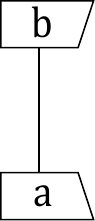}
\qquad \forall~a,b\in A
&&
\Rightarrow~~~
\includegraphics[scale=.6,valign=c]{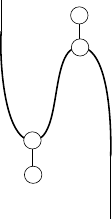}
~~=~~
\includegraphics[scale=.6,valign=c]{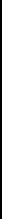}
\end{calign}
We also have the following equation, following from the equation $\braket{c|a^*b} = \phi(c^*a^*b) = \phi((ac)^*b) = \braket{ac|b}$:
\begin{calign}\label{eq:frobpfcomult}
\includegraphics[scale=.6,valign=c]{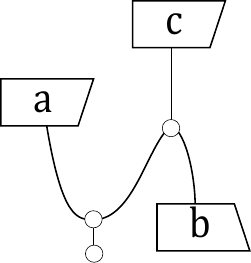}
~~=~~
\includegraphics[scale=.6,valign=c]{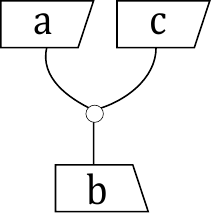}
\qquad \forall~a,b,c\in A
&&
\Rightarrow~~~
\includegraphics[scale=.6,valign=c]{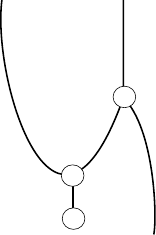}
~~=~~
\includegraphics[scale=.6,valign=c]{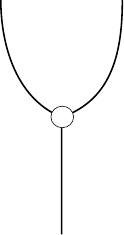}
\end{calign}
From this we obtain the following additional equation:
\begin{calign}\label{eq:frobpfmult}
\includegraphics[scale=.6,valign=c]{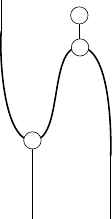}
~~=~~
\includegraphics[scale=.6,valign=c]{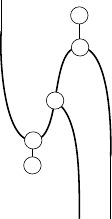}
~~=~~
\includegraphics[scale=.6,valign=c]{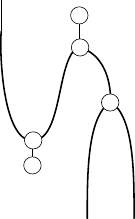}
~~=~~
\includegraphics[scale=.6,valign=c]{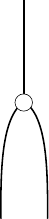}
\end{calign}
Then for any elements $\ket{a},\ket{b},\ket{c},\ket{d} \in A$ and some orthonormal basis $\{\ket{k}\}$ of $A$:
\begin{calign}\nonumber
\includegraphics[scale=.6,valign=c]{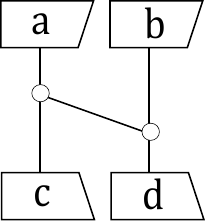}
~~=~~
\sum_{k}
\includegraphics[scale=.6,valign=c]{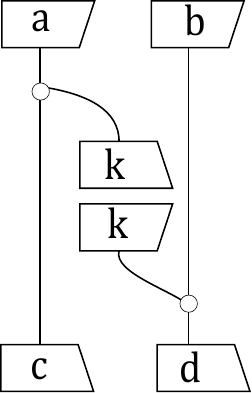}
~~=~~\sum_{k}
\includegraphics[scale=.6,valign=c]{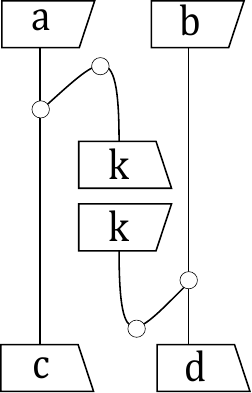}
~~=~~
\includegraphics[scale=.6,valign=c]{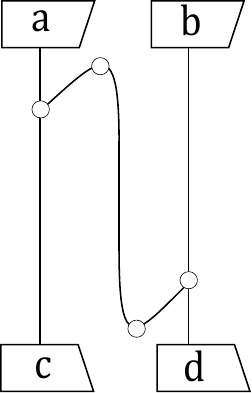}
~~=~~
\includegraphics[scale=.6,valign=c]{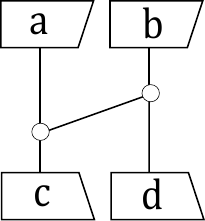}
\end{calign}
Here the first equality is by insertion of a resolution of the identity $\sum_k \ket{k}\bra{k}$; the second equality is by~\eqref{eq:frobpfcomult} and~\eqref{eq:frobpfmult}; the third equality is by removal of the resolution of the identity; and the fourth equality is by the dagger of~\eqref{eq:frobpfsnake}. The Frobenius equation is therefore shown.

Finally, we must show that $\phi = u^{\dagger}$. For this we need only observe that for all $a \in A$, $\braket{1|a} = \phi(1^*a) = \phi(a)$.
\end{proof}
\noindent
The converse also holds.
\begin{proposition}[{\cite[Lem. 2.2]{Neshveyev2018}}]\label{prop:frobtocstar}
Let $[A,m,u]$ be a Frobenius algebra in $\Hilb$. Then there is a unique involution on $A$, defined as in~\eqref{eq:frobinvol}, such that $A$ becomes a f.d. $C^*$-algebra with faithful positive linear functional $\phi=u^{\dagger}: A \to \mathbb{C}$ such that $\braket{x|y} = \phi(x^{*}y)$.
\end{proposition} 
\noindent
We have therefore identified Frobenius algebras in $\Hilb$ with f.d. $C^*$-algebras equipped with a faithful positive linear functional. The maps between Frobenius algebras we have already considered in Definition~\ref{def:starhomcohom} correspond to their equivalent notions for finite-dimensional $C^*$-algebras. 
\begin{proposition}\label{prop:starhomscstaralg}
In $\Hilb$:
\begin{itemize} 
\item A $*$-homomorphism of Frobenius algebras is a $*$-homomorphism of f.d. $C^*$-algebras. 
\item A $*$-isomorphism of Frobenius algebras is a $*$-isomorphism of f.d. $C^*$-algebras.
\item A unitary $*$-isomorphism of Frobenius algebras is a $*$-isomorphism of f.d. $C^*$-algebras preserving the faithful positive linear functional $u^{\dagger}$.
\end{itemize}
\end{proposition}
\noindent
We have shown a correspondence between Frobenius algebras in $\Hilb$ and pairs of an f.d. $C^*$-algebras and a faithful positive linear functional. We now consider a Frobenius-algebraic formulation of channels between these f.d. $C^*$-algebras. We will consider a channel to be a completely positive linear map preserving the chosen functional. Of course, it is common in physics to restrict to the case where the functional is a trace (see Remark~\ref{rem:trace}).
\begin{definition}
Let $A,B$ be Frobenius algebras in a rigid $C^*$-tensor category $\mathcal{C}$ and let $f: A \to B$ be a morphism. We say that $f$ satisfies the \emph{CP condition}, or is a \emph{CP morphism}, when there exists an object $X$ and a morphism $g: A \otimes B \to X$ such that the following equation holds:
\begin{calign}\label{eq:cpcond}
\includegraphics[scale=.8,valign=c]{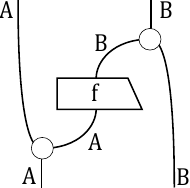}
~~=~~
\includegraphics[scale=.8,valign=c]{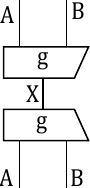}
\end{calign}
In other words, the morphism on the LHS of~\eqref{eq:cpcond} is a positive element of the f.d. $C^*$-algebra $\End(A \otimes B)$.
\end{definition}
\begin{theorem}[{\cite[Thm. 7.18]{Heunen2019}}]\label{thm:cpmap}
Let $A,B$ be Frobenius algebras in $\Hilb$. A morphism $N: A \to B$ is completely positive as a map between f.d. $C^*$-algebras iff it satisfies the CP condition. 
\end{theorem}
\noindent
We therefore obtain the following definition of a channel. 
\begin{definition}\label{def:chan}
Let $A, B$ be Frobenius algebras and $f: A \to B$ a morphism in a rigid $C^*$-tensor category. We say that $f$ is a \emph{channel} when it satisfies the CP condition~\eqref{eq:cpcond} and preserves the counit:
\begin{calign}\label{eq:tpcond}
\includegraphics[scale=.8,valign=c]{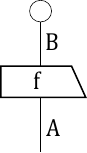}
~~=~~
\includegraphics[scale=.8,valign=c]{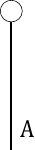}
\end{calign}
\end{definition}
\begin{corollary}\label{cor:chan}
Let $A,B$ be Frobenius algebras in $\Hilb$. A linear map $N: A \to B$ is a completely positive functional-preserving map between f.d. $C^*$-algebras precisely when it is a channel in the sense of Definition~\ref{def:chan}.
\end{corollary}
\noindent
Using this description, we can quickly show that $*$-cohomomorphisms between separable Frobenius algebras are channels. 
\begin{lemma}\label{lem:starcohom}
Let $A,B$ be separable Frobenius algebras in $\Hilb$. Any $*$-cohomomorphism $f: A \to B$ is a channel.
\end{lemma}
\begin{proof}
The functional-preservation condition~\eqref{eq:tpcond} is immediately satisfied, by~\eqref{eq:cohomo}. We show the CP condition:
\begin{calign}\nonumber
\includegraphics[scale=.8,valign=c]{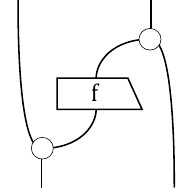}
~~=~~
\includegraphics[scale=.8,valign=c]{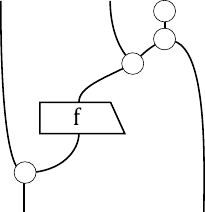}
~~=~~
\includegraphics[scale=.8,valign=c]{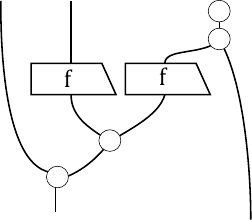}
~~=~~
\includegraphics[scale=.8,valign=c]{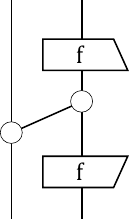}
~~=~~
\includegraphics[scale=.8,valign=c]{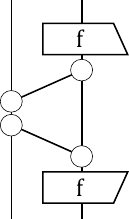}
\end{calign}
Here the first equality is by the Frobenius equation~
\eqref{eq:Frobenius} and counitality; the second equality is by the first $*$-cohomomorphism condition~\eqref{eq:cohomo}; the third equality is by the third $*$-cohomomorphism condition~\eqref{eq:cohomo}; and the final equality is by separability of the Frobenius algebra and~\eqref{eq:Frobenius}.
\end{proof}

\begin{example}[The algebra $B(H)$ and the maximally entangled state]\label{ex:bh}
Let $H$ be a Hilbert space of dimension $d$. 
We define a separable Frobenius algebra on the Hilbert space $H \otimes H$ with the following multiplication $m: (H \otimes H) \otimes (H \otimes H) \to H \otimes H$ and unit $u: \mathbb{C} \to H \otimes H$:
\begin{align*}
\frac{1}{\sqrt{d}}~\includegraphics[scale=1,valign=c]{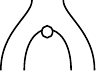}
&&
\sqrt{d}~\includegraphics[scale=1,valign=c]{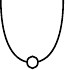}
\end{align*}
By Proposition~\ref{prop:frobtocstar}, this is a finite-dimensional $C^*$-algebra with involution~\eqref{eq:frobinvol} and faithful trace $u^{\dagger}$. We will now show that there is a trace-preserving $*$-isomorphism between this f.d. $C^*$-algebra and the algebra $B(H)$ equipped with the \emph{separable trace} $d \Tr: B(H) \to \mathbb{C}$, where $\Tr$ is the matrix trace. 

We claim that the following linear map is a $*$-isomorphism:
\begin{align}\nonumber
f: B(H) &\to H \otimes H \\\label{eq:endhiso}
X &\mapsto \sqrt{d}(X \otimes \mathbbm{1})(\sum_i \ket{i} \otimes \ket{i}) 
\end{align}
Multiplicativity is seen as follows:
\begin{calign}\nonumber
\sqrt{d}~~
\includegraphics[scale=.8,valign=c]{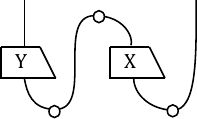}
~~=~~
\sqrt{d}~~
\includegraphics[scale=.8,valign=c]{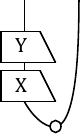}
\end{calign}
That the unit is preserved is clear from the definition. To see that the involution is preserved:
\begin{calign}\nonumber
\sqrt{d}~~
\includegraphics[scale=.8,valign=c]{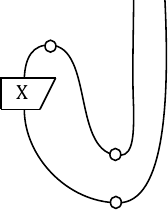}
~~=~~
\sqrt{d}~~
\includegraphics[scale=.8,valign=c]{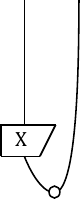}
\end{calign}
To see that it is an injection, observe that any element of $B(H)$ can be recovered from its image:
\begin{calign}\nonumber
\sqrt{d}~~
\includegraphics[scale=.8,valign=c]{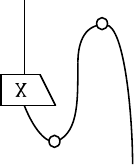}
~~=~~
\sqrt{d}~~
\includegraphics[scale=.8,valign=c]{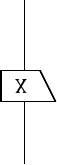}
\end{calign}
To see that it is a surjection, observe that any element of $H \otimes H$ can be written as $\sum_{ij} \lambda_{ij} \ket{i} \otimes \ket{j}$. This is the image of the element $\frac{1}{\sqrt{d}}\sum_{ij}\lambda_{ij} \ket{i}\bra{j}$ of $B(H)$.
To see that the trace is preserved:
\begin{calign}\nonumber
d~~\includegraphics[scale=.8,valign=c]{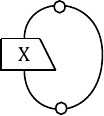}
~~=~~
d \Tr(X)
\end{calign}
We now consider the representation of the maximally entangled state of $H \otimes H$. Let us consider the f.d. $C^*$-algebra $B(H \otimes H)$ with the separable trace, defined as above as a Frobenius algebra on $(H \otimes H) \otimes (H \otimes H)$. Since the maximally entangled state of $H \otimes H$ is just a normalised version of the cup of the self-duality of $H$, the maximally entangled state of $B(H \otimes H)$ has the following form as an element of the Frobenius algebra:
\begin{align*}
\frac{1}{d^2}~
\includegraphics[scale=.8,valign=c]{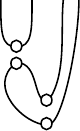}
~~=~~
\frac{1}{d^2}~
\includegraphics[scale=.8,valign=c]{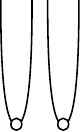}
\end{align*}
Using the canonical unitary $*$-isomorphism $B(H \otimes H) \cong B(H) \otimes B(H)$, we obtain the following expression for the maximally entangled state as an element of the Frobenius algebra $B(H) \otimes B(H)$:
\begin{align}
\label{eq:maxentstate}
\frac{1}{d^2}~~\includegraphics[scale=.8,valign=c]{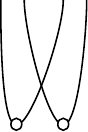}
\end{align}
\end{example}

\subsection{Finite-dimensional $G$-$C^*$-algebras and covariant channels}
\label{sec:gcstaralg}
We now generalise to the case where operations are covariant for the action of a compact quantum group $G$. 
We first recall from~\cite[\S{}2.3]{Neshveyev2018} the following completely general characterisation of a finite-dimensional $G$-$C^*$-algebra $A$ equipped with a $G$-invariant positive linear functional.
\begin{definition}\label{def:generalgcstaralg}
Let $[A,m,u]$ be a Frobenius algebra in $\Hilb$ and let $A_G$ be a compact quantum group algebra. We say that a linear map $\rho: A \to A \otimes A_G$ is a \emph{coaction} of $A_G$ on $A$ if the pair $(A,\rho)$ is a unitary corepresentation (Definition~\ref{def:unitarycorep}) and $m:A \otimes A \to A, u: \mathbb{C} \to A$ are intertwiners, i.e.:
\begin{calign}
\includegraphics[scale=.8,valign=c]{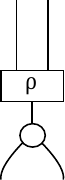}
~~=~~
\includegraphics[scale=.8,valign=c]{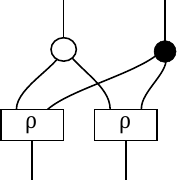}
&
\includegraphics[scale=.8,valign=c]{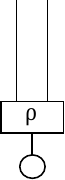}
~~=~~
\includegraphics[scale=.8,valign=c]{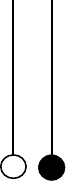}
\end{calign}
Here the black vertex represents the multiplication and unit of the Hopf $*$-algebra $A_G$. We call $([A,m,u],\rho)$ a \emph{f.d. $G$-$C^*$-algebra with functional}. (Here the $G$-invariant faithful positive linear functional is the counit $u^{\dagger}: A \to \mathbb{C}$.)

We say that a linear map between f.d. $G$-$C^*$-algebras with functional is \emph{covariant} if it is an intertwiner of corepresentations (Definition~\ref{def:unitarycorep}).
\end{definition}
\noindent
Definition~\ref{def:generalgcstaralg} is just a concrete way of saying that a f.d. $G$-$C^*$-algebra equipped with a $G$-invariant faithful positive linear functional is a Frobenius algebra in $\Rep(G)$. Following Proposition~\ref{prop:starhomscstaralg} we can identify $*$-homomorphisms,  $*$-isomorphisms, and unitary $*$-isomorphisms of Frobenius algebras in $\Rep(G)$ with covariant $*$-homomorphisms, covariant $*$-isomorphisms, and functional-preserving covariant $*$-isomorphisms of f.d. $G$-$C^*$-algebras respectively. 

A f.d. $G$-$C^*$-algebra therefore corresponds to a $*$-isomorphism class of Frobenius algebras in $\Rep(G)$. In what follows we will use a result from~\cite{Neshveyev2018} to uniquely fix a functional for each f.d. $G$-$C^*$-algebra.
\begin{lemma}\label{lem:uniquespecialfunctional}
Every f.d. $G$-$C^*$-algebra admits a unique $G$-invariant faithful positive linear functional such that the corresponding Frobenius algebra in $\Rep(G)$ is separable, standard and satisfies $u^{\dagger} u = \dim_q(A)$. 
\end{lemma}
\begin{proof}
This is just~\cite[Thm. 2.11]{Neshveyev2018} with a different normalisation. There it was shown that there is a unique $G$-invariant state on any f.d. $G$-$C^*$-algebra such that the corresponding Frobenius algebra $A$ is standard and satisfies $u^{\dagger} u= 1$ and  $mm^{\dagger}=\dim_q(A)\id_A$. Using a scalar $*$-isomorphism of Frobenius algebras ($u \mapsto \sqrt{\dim_q(A)} u$ and $m \mapsto \frac{1}{\sqrt{\dim_q(A)}} m$) we obtain the desired unique $G$-invariant faithful positive linear functional.
\end{proof}
\noindent 
We call the unique $G$-invariant faithful positive linear functional of Lemma~\ref{lem:uniquespecialfunctional} the \emph{canonical $G$-invariant functional} on $A$. This is just a different normalisation of the canonical invariant state of~\cite[Thm. 2.11]{Neshveyev2018}. 
\begin{remark}\label{rem:trace}
It is common in physics to assume that the functional is a trace. In the case of a general compact quantum group action there may be no $G$-invariant trace on a f.d. $G$-$C^*$-algebra. We discuss this in Appendix~\ref{sec:app}. In particular, we show that, whenever a $G$-invariant trace exists, the canonical invariant functional of Lemma~\ref{lem:uniquespecialfunctional} is tracial. We also provide a necessary and sufficient condition for the existence of a $G$-invariant trace; namely, that the quantum dimension $\dim_q(A)$ should be equal to the ordinary dimension $\dim(A)$. In this case we show in that the canonical functional is precisely the \emph{separable trace} on the underlying f.d. $C^*$-algebra~\cite[Thm. 4.6]{Vicary2011}. In particular, this implies that the canonical invariant functional is always tracial whenever the compact quantum group is of \emph{Kac type}~\cite[Prop. 1.7.9]{Neshveyev2013} (this includes all ordinary compact groups). 
\end{remark}
\noindent
Let $\mathcal{C}$ be a rigid $C^*$-tensor category. Recall that, by Tannaka duality (Theorem~\ref{thm:tannaka}), when $U: \mathcal{C} \to \Hilb$ is a fibre functor there is a compact quantum group $G$ associated to the pair $(\mathcal{C},U)$ such that $\mathcal{C} \simeq \Rep(G)$. One therefore expects that f.d. $G$-$C^*$-algebras can be identified with separable standard Frobenius algebras ($\F$s) in $\mathcal{C}$; in what follows we make this precise.
\begin{proposition}\label{prop:imisgcalg}
Let $\mathcal{C}$ be a rigid $C^*$-tensor category and $U: \mathcal{C} \to \Hilb$ be a fibre functor. Let $[A,m,u]$ be an $\F$ in $\mathcal{C}$. Then $[U(A),U(m) \circ \mu_{A,A},U(u) \circ \upsilon]$ is a separable Frobenius algebra in $\Hilb$:
\begin{calign}\nonumber
\includegraphics[scale=.8,valign=c]{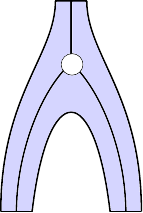}
&
\includegraphics[scale=.8,valign=c]{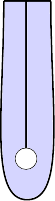}
\\ \nonumber
U(m) \circ \mu_{A,A} 
&
U(u) \circ \upsilon
\end{calign}
Let $G$ be the compact quantum group obtained by Tannaka reconstruction from the pair $(\mathcal{C},U)$, and let $\rho_{U,A}: U(A) \to U(A) \otimes A_G$ be the corresponding corepresentation of $A_G$ on $U(A)$. Then $([U(A),U(m) \circ \mu_{A,A},U(u) \circ \upsilon],\rho_{U,A})$ is a f.d. $G$-$C^*$-algebra with functional.
\end{proposition}
\begin{proof}
For the first statement, to show that $[U(A),U(m) \circ \mu_{A,A},U(u) \circ \upsilon]$ is a \F{} it is easy to check that the equations of a separable Frobenius algebra are preserved under a unitary monoidal functor. For example, separability follows from unitarity of the fibre functor (Definition~\ref{def:fibre}) and separability of $[A,m,u]$:
\begin{calign}\nonumber
\includegraphics[scale=.8,valign=c]{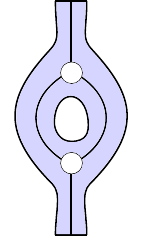}
~~=~~
\includegraphics[scale=.8,valign=c]{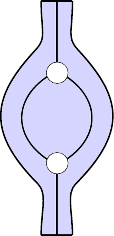}
~~=~~
\includegraphics[scale=.8,valign=c]{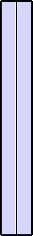}
\end{calign}
For the second statement, we already know that $\rho_{U,A}: U(A) \to U(A) \otimes A_G$ is a unitary corepresentation, so we need only show that $U(m)\circ \mu_{A,A}$ and $U(u) \circ \upsilon$ are intertwiners. This follows from the properties of the representations $\{\rho_{U,X}\}_{X \in \Obj(\mathcal{C})}$ listed in Section~\ref{sec:tannaka}. To see that $U(m) \circ \mu_{A,A}$ is an intertwiner:
\begin{calign}\nonumber
\includegraphics[scale=.9,valign=c]{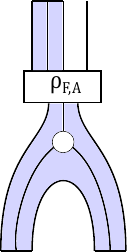}
~~=~~
\includegraphics[scale=.9,valign=c]{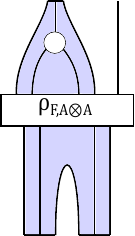}
~~=~~
\includegraphics[scale=.9,valign=c]{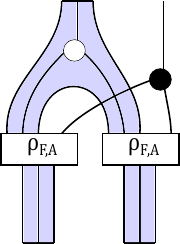}
\end{calign}
Here the first equality is by naturality~\eqref{eq:rhonat} of $\{\rho_{U,X}\}$ and the second is by monoidality~\eqref{eq:rhomon} of $\{\rho_{U,X}\}$.

To see that $U(u) \circ \upsilon$ is an intertwiner:
\begin{calign}\nonumber
\includegraphics[scale=.9,valign=c]{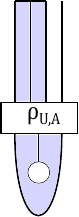}
~~=~~
\includegraphics[scale=.9,valign=c]{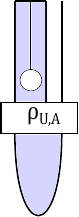}
~~=~~
\includegraphics[scale=.9,valign=c]{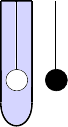}
\end{calign}
Here the first equality is by naturality~\eqref{eq:rhonat} and the second equality is by monoidality~\eqref{eq:rhomonunit}.
\end{proof}
\begin{notation}
To be concise, for any $\F$ $[A,m,u]$ in $\mathcal{C}$ we write $U([A,m,u])$ for the f.d. $G$-$C^*$-algebra $([U(A),U(m) \circ \mu_{A,A},U(u) \circ \upsilon],\rho_{U,A})$. 
\end{notation}
\noindent
We now show that, up to covariant $*$-isomorphism, all f.d. $G$-$C^*$-algebras can be obtained in this way.
\begin{proposition}
Let $\mathcal{C}$ be a rigid $C^*$-tensor category, let $U: \mathcal{C} \to \Hilb$ be a fibre functor, and let $G$ be the compact quantum group obtained by Tannaka reconstruction for the pair $(\mathcal{C},U)$. 

Let $([A,m,u],\rho)$ be a f.d. $G$-$C^*$-algebra with functional. Then there exists a $\F$ $[\bar{A},\bar{m},\bar{u}]$ in $\mathcal{C}$ such that $U([\bar{A},\bar{m},\bar{u}])$ is covariantly $*$-isomorphic to $([A,m,u],\rho)$.
\end{proposition}
\begin{proof}
Recall the commuting diagram of functors~\eqref{diag:tannaka}. We first observe that, since $(A,\rho)$ is a unitary corepresentation of $G$ and $E$ is a unitary equivalence, there exists an object $\tilde{A}$ of $\mathcal{C}$ and a unitary intertwiner $f: (U(\tilde{A}),\rho_{U,\tilde{A}}) \to (A,\rho)$. We pull back the algebra structure of $[A,m,u]$ along $f$ to obtain a f.d. $G$-$C^*$-algebra $([U(\tilde{A}),f^{\dagger} \circ m \circ (f \otimes f),f^{\dagger} \circ u],\rho_{U,\tilde{A}})$, so that $f$ is a covariant $*$-isomorphism. It is easy to show by fullness and unitarity of the monoidal equivalence $E$ that there exist morphisms $\tilde{m}: \tilde{A} \otimes \tilde{A} \to \tilde{A}$ and $\tilde{u}: \mathbbm{1} \to \tilde{A}$ in $\mathcal{C}$ such that $([U(\tilde{A}),f^{\dagger} \circ m \circ (f \otimes f),f^{\dagger} \circ u],\rho_{U,\tilde{A}}) = U([\tilde{A},\tilde{m},\tilde{u}])$. Finally, by Lemma~\ref{lem:uniquespecialfunctional} there is a standard separable Frobenius algebra $[\bar{A},\bar{m},\bar{u}]$ in $\mathcal{C}$ which is covariantly $*$-isomorphic to $[\tilde{A},\tilde{m},\tilde{u}]$.
\end{proof}
\noindent
Having identified f.d. $G$-$C^*$-algebras with $\F$s in $\mathcal{C}$, we now turn to covariant channels. 
\begin{definition}
Let $([A_1,m_1,u_1],\rho_1)$, $([A_2,m_2,u_2],\rho_2)$ be f.d. $G$-$C^*$-algebras with functional. We say that a CP morphism $f: [A_1,m_1,u_1] \to{} [A_2,m_2,u_2]$ is \emph{covariant} if it is also an intertwiner $(A_1,\rho_1) \to (A_2,\rho_2)$.
\end{definition}
\noindent
We first show that a CP morphism between $\F$s in $\mathcal{C}$ induces a covariant CP morphism between the corresponding f.d. $G$-$C^*$-algebras.
\begin{proposition}\label{prop:imiscovchan}
Let $\mathcal{C}$ be a rigid $C^*$-tensor category, and let $U: \mathcal{C} \to \Hilb$ be a fibre functor. Let $[A_1,m_1,u_1]$, $[A_2,m_2,u_2]$ be $\F$s in $\mathcal{C}$ and let $f: [A_1,m_1,u_1] \to{} [A_2,m_2,u_2]$ be a CP morphism in $\mathcal{C}$. Then $$U(f): U([A_1,m_1,u_1]) \to{} U([A_2,m_2,u_2])$$ is a covariant CP morphism.
If $f$ is trace-preserving then $U(f)$ is a covariant channel.
\end{proposition}
\begin{proof}
We know that $U(f)$ is an intertwiner by naturality~\eqref{eq:rhonat} of $\{\rho_{U,X}\}_{X \in \Obj(\mathcal{C})}$.
We show that $U(f)$ is a CP map. Indeed, the CP condition~\eqref{eq:cpcond} is preserved under the unitary monoidal functor $U$: 
\begin{calign}\nonumber
\includegraphics[scale=.9,valign=c]{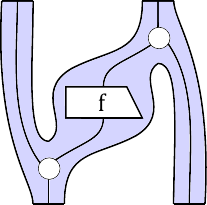}
~~=~~
\includegraphics[scale=.9,valign=c]{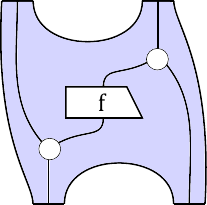}
~~=~~
\includegraphics[scale=.9,valign=c]{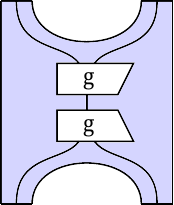}
\end{calign}
The map $U(f)$ is therefore covariant CP.
Finally, the trace-preservation condition~\eqref{eq:tpcond} is preserved under the functor $U$:
\begin{calign}\nonumber
\includegraphics[scale=.9,valign=c]{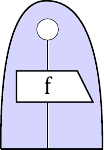}
~~=~~
\includegraphics[scale=.9,valign=c]{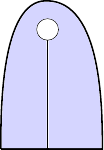}
\end{calign}
\end{proof}
\noindent
We now show that all covariant CP maps and channels are obtained as the image of CP morphisms in $\mathcal{C}$.
\ignore{
\begin{lemma}\label{lem:pos*hom}
Let $A,B$ be f.d. $C^*$-algebras and let $\pi:A \to B$ be a $*$-homomorphism. Then a element $\pi(a)$ is self-adjoint and positive in $B$ iff $a$ is self-adjoint and positive in $A$.
\end{lemma}
\begin{proof}
The self-adjoint part is clear. The positivity part follows from the characterisation of positivity of a self-adjoint element in  terms of positivity of the spectrum~\cite[Thm. 1.6.5]{Putnam}. In one direction, if $a$ is positive in $A$ --- i.e. $a = b b^*$ for some $b \in A$ --- then $f(a) = f(b) f(b)^*$ and so $f(a)$ is positive in $B$. In the other direction, if $f(a)$ is positive in $B$, then $\lambda >0$ for all $\lambda$ s.t. $\lambda1_B-f(a)$ is invertible in $B$. Consider $\lambda 1_A - a \in A$. If this is invertible then $f(\lambda 1_A-a)=\lambda 1_B - f(a)$ is invertible in $B$: thus $\lambda>0$. 
\end{proof}}

\begin{proposition}
Let $\mathcal{C}$ be a rigid $C^*$-tensor category, and let $U: \mathcal{C} \to \Hilb$ be a fibre functor. 
Let $[A_1,m_1,u_1]$,$[A_2,m_2,u_2]$ be $\F$s in $\mathcal{C}$.
Then for every covariant CP map $$f: U([A_1,m_1,u_1]) \to U([A_2,m_2,u_2])$$ there is a unique CP morphism $\tilde{f}: [A_1,m_1,u_1] \to{} [A_2,m_2,u_2]$ in $\mathcal{C}$ such that $U(\tilde{f}) = f$.

If $f$ is trace-preserving then $\tilde{f}$ is trace-preserving also.
\end{proposition}
\begin{proof}
Recall the commuting triangle~\eqref{diag:tannaka}. By covariance of $f$, there is a unique morphism $\hat{f}$ in $\Rep(G)$ such that $F(\hat{f}) = f$, where $F$ is the canonical fibre functor. Now since $U=F \circ E$ and $E$ is an equivalence, there exists a unique morphism $\tilde{f}$ in $\mathcal{C}$ such that $U(\tilde{f}) = f$.
 
We now need to show that $\tilde{f}$ obeys the CP condition in $\mathcal{C}$. We know that $f$ satisfies the CP condition~\eqref{eq:cpcond} in $\Hilb$; this is precisely to say that the morphism
\begin{calign}\nonumber
\includegraphics[scale=.9,valign=c]{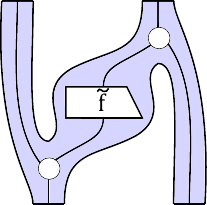}
\end{calign} 
is a positive element in the f.d. $C^*$-algebra $\End_{\Hilb}(U(A) \otimes U(B))$. 
This in turn implies (by pre- and post-composition with $\mu_{A,B}^{\dagger}$ and $\mu_{A,B}$ respectively) that the following morphism is a positive element in $\End_{\Hilb}U(A \otimes B)$:
\begin{calign}\nonumber
\includegraphics[scale=.9,valign=c]{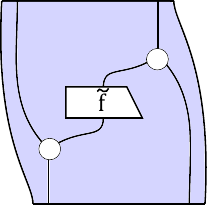}
\end{calign}
But the fibre functor $U$ is faithful and unitary; in particular, the induced map $\End_{\mathcal{C}}(A \otimes B) \to \End_{\Hilb}(U(A \otimes B))$ is an injective $*$-homomorphism. It follows from spectral permanence that the following is a positive element of $\End_{\mathcal{C}}(A \otimes B)$:
\begin{calign}\nonumber
\includegraphics[scale=.9,valign=c]{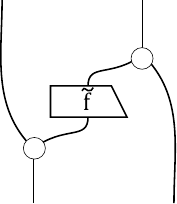}
\end{calign}
This is precisely to say that $\tilde{f}$ is a CP morphism in $\mathcal{C}$.

Finally, we consider the trace-preservation condition. If the covariant CP map $f$ is trace-preserving, then we have the following sequence of implications:
\begin{calign}\nonumber
\includegraphics[scale=.9,valign=c]{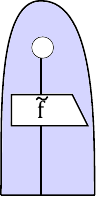}
~~=~~
\includegraphics[scale=.9,valign=c]{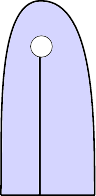}
~~~~\Rightarrow~~~~
\includegraphics[scale=.9,valign=c]{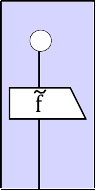}
~~=~~
\includegraphics[scale=.9,valign=c]{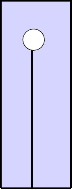}
~~~~\Rightarrow~~~~
\includegraphics[scale=.9,valign=c]{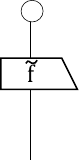}
~~=~~
\includegraphics[scale=.9,valign=c]{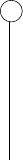}
\end{calign}
Here the first implication is by postcomposition with $\upsilon$, and the second is by faithfulness of the functor $U$.
The last equality precisely states that $\tilde{f}: [A_1,m_1,u_1] \to{} [A_2,m_2,u_2]$ is trace-preserving.
\end{proof}
\noindent
Altogether, what we have proven is an equivalence of categories. We formulate this equivalence now. Let $\mathcal{C}$ be a rigid $C^*$-tensor category, let $U: \mathcal{C} \to \Hilb$ be a fibre functor, and let $G$ be the compact quantum group recovered from the pair $(\mathcal{C},U)$ by Tannaka duality. Let $\CP(\mathcal{C})$ be the category whose objects are standard separable Frobenius algebras in $\mathcal{C}$ and whose morphisms are CP morphisms between them. Let $\Chan(\mathcal{C})$ be the subcategory whose morphisms additionally preserve the counit. Let $\CP(G)$ be the category whose objects are f.d. $G$-$C^*$-algebras and whose morphisms are covariant CP maps. Let $\Chan(G)$ be the subcategory whose morphisms additionally preserve the canonical $G$-invariant functional.
\begin{theorem}\label{thm:equivalence}
The functors $\CP(\mathcal{C}) \to \CP(G)$ and $\Chan(\mathcal{C}) \to \Chan(G)$ defined on objects by $[A,m,u] \mapsto U([A,m,u])$ and on morphisms by $f \mapsto U(f)$ are equivalences of categories. 
\end{theorem}

\section{Entanglement-symmetries of covariant channels}
\label{sec:thm}
We can finally prove the construction of entanglement-symmetries discussed in the introduction. 

For the remainder of this section we fix a compact quantum group $G$ and let $F: \Rep(G) \to \Hilb$ be the canonical fibre functor. Let $F': \Rep(G) \to \Hilb$ be some other fibre functor , and let $G'$ be the compact quantum group associated to it by Tannaka duality. Let $(\alpha,H_e): F \to F'$ be a UPT, and let $d:=\dim(H_e)$.

Recall the definition of a $*$-cohomomorphism (Definition~\ref{def:starhomcohom}) between Frobenius algebras, and that a $*$-cohomomorphism between separable Frobenius algebras in $\Hilb$ is in particular a channel (Lemma~\ref{lem:starcohom}). Recall also the Frobenius algebra structure of $B(H_e) \cong H \otimes H$ (Example~\ref{ex:bh}).
\begin{theorem}\label{thm:starcohoms}
Let $[A,m,u]$ be an $\F$ in $\Rep(G)$. Then we obtain:
\begin{itemize}
\item A $*$-cohomomorphism 
$$u_{[A,m,u]}: F([A,m,u]) \otimes B(H_e) \to{} F'([A,m,u])$$
defined as follows:
\begin{calign}\label{eq:udef}
\includegraphics[scale=0.9,valign=c]{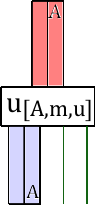}
~~:=~~
\sqrt{d}
\includegraphics[scale=0.9,valign=c]{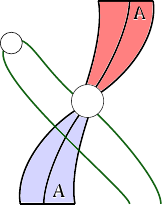}
\end{calign}
\item A $*$-cohomomorphism 
$$v_{[A,m,u]}: F'([A,m,u]) \otimes B(H_e) \to{} F([A,m,u])$$
defined as follows:
\begin{calign}\label{eq:vdef}
\includegraphics[scale=0.9,valign=c]{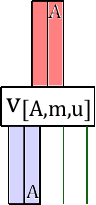}
~~:=~~
\sqrt{d}
\includegraphics[scale=0.9,valign=c]{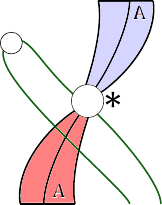}
\end{calign}
\end{itemize}
These $*$-cohomomorphisms obey the following \emph{entanglement-invertibility} equations with respect to the maximally entangled state $\Psi: \mathbb{C} \to B(H) \otimes B(H)$:
\begin{calign}\label{eq:entinvertibility}
\includegraphics[scale=0.9,valign=c]{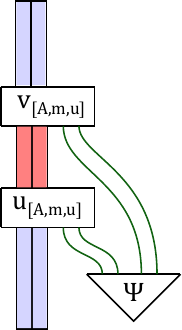}
~~=~~
\includegraphics[scale=0.9,valign=c]{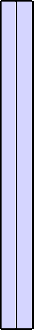}
&&
\includegraphics[scale=0.9,valign=c]{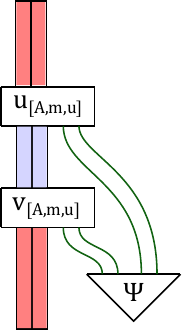}
~~=~~
\includegraphics[scale=0.9,valign=c]{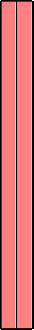}
\end{calign}
These $*$-cohomomorphisms are also \emph{natural} for CP morphisms $f: [A_1,m_1,u_1] \to{} [A_2,m_2,u_2]$:
\begin{calign}\label{eq:starcohomsnat}
\includegraphics[scale=0.9,valign=c]{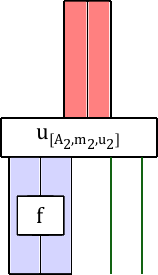}
~~=~~
\includegraphics[scale=0.9,valign=c]{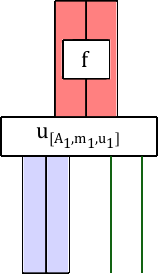}
&&
\includegraphics[scale=0.9,valign=c]{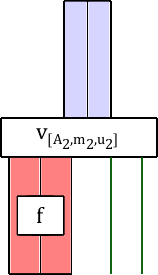}
~~=~~
\includegraphics[scale=0.9,valign=c]{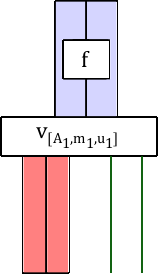}
\end{calign}
\end{theorem}
\begin{proof}
We begin by showing that the maps $\{u_{[A,m,u]}\}$ and $\{v_{[A,m,u]}\}$ are $*$-cohomomorphisms. Recall the definition of the multiplication and unit of $B(H)$ (Example~\ref{ex:bh}) and the tensor product of $\F$s (Definition~\ref{def:tensorprodfrob}). We show the proof for $\{u_{[A,m,u]}\}$; the proof for $\{v_{[A,m,u]}\}$ is similar. For the first equation of~\eqref{eq:cohomo}:
\begin{calign}\nonumber
\sqrt{d}
\includegraphics[scale=0.9,valign=c]{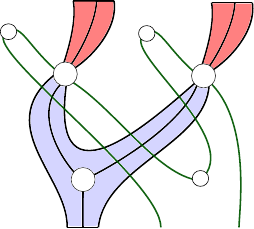}
~~=~~
\sqrt{d}~
\includegraphics[scale=0.9,valign=c]{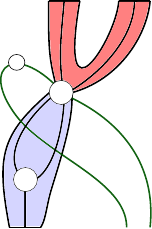}
~~=~~
\sqrt{d}~~
\includegraphics[scale=0.9,valign=c]{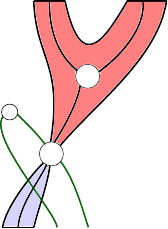}
\end{calign}
Here the first equality is by a snake equation and monoidality of the UPT $\alpha$~\eqref{eq:pntmonmon}; the second equation is by naturality of $\alpha$~\eqref{eq:pntmonnat}. For the second equation of~\eqref{eq:cohomo}:
\begin{calign}\nonumber
\sqrt{d}
\includegraphics[scale=0.9,valign=c]{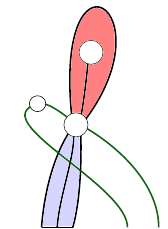}
~~=~~
\sqrt{d}~~
\includegraphics[scale=0.9,valign=c]{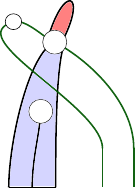}
~~=~~
\sqrt{d}~~
\includegraphics[scale=0.9,valign=c]{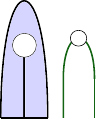}
\end{calign}
Here the first equality is by naturality of the UPT $\alpha$~\eqref{eq:pntmonnat} and the second equality is by monoidality of $\alpha$~\eqref{eq:pntmonmonunit} and unitarity of the functor $F'$.
For the third equation of~\eqref{eq:cohomo}:
\begin{calign}\nonumber
\sqrt{d}~~
\includegraphics[scale=0.9,valign=c]{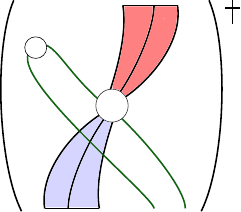}
~~=~~
\sqrt{d}
\includegraphics[scale=0.9,valign=c]{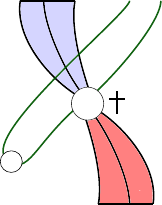}
~~=~~
\sqrt{d}
\includegraphics[scale=0.9,valign=c]{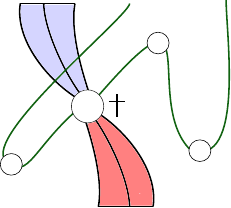}
\\\nonumber~~=~~
\sqrt{d}~
\includegraphics[scale=0.9,valign=c]{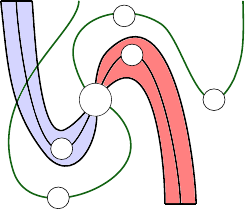}
~~=~~
\sqrt{d}~~
\includegraphics[scale=0.9,valign=c]{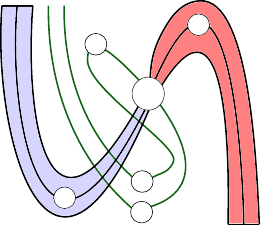}
\end{calign}
Here the first equality is by the `horizontal flip' calculus of the dagger; the second equality is by a snake equation; the third equality is by the definition of the dual UPT~\eqref{eq:dualpnt};\footnote{We have simplified the proof by assuming that the chosen dual of $A$ in the pivotal dagger category $\Rep(G)$ is $A$ itself under the duality corresponding to the Frobenius algebra $[A,m,u]$. The proof can easily be made fully general using uniqueness of standard duals up to unitary isomorphism.} and the fourth equality is by~\eqref{eq:untangling}. In the final diagram we recognise the cup~\eqref{eq:cupcapfrob} of $F([A,m,u]) \otimes B(H_e)$ and the cap~\eqref{eq:cupcapfrob} of $F'([A,m,u])$.

We now show the entanglement-invertibility equations~\eqref{eq:entinvertibility}. Recall the expression~\eqref{eq:maxentstate} for the channel  $\Psi: \mathbb{C} \to B(H) \otimes B(H)$ initialising the maximally entangled state. We prove the first equation of~\eqref{eq:entinvertibility}:
\begin{calign}\nonumber
\frac{1}{d}~~
\includegraphics[scale=0.9,valign=c]{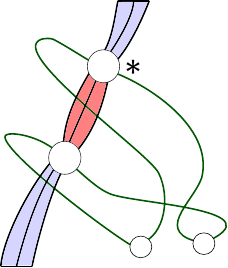}
~~=~~
\frac{1}{d}~~
\includegraphics[scale=0.9,valign=c]{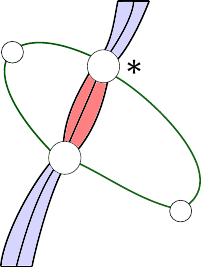}
~~=~~
\frac{1}{d}~~
\includegraphics[scale=0.9,valign=c]{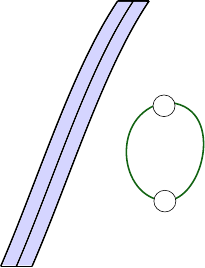}
~~=~~
\includegraphics[scale=0.9,valign=c]{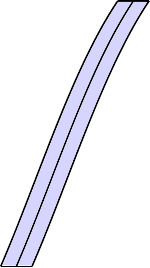}
\end{calign}
Here the first equality is by~\eqref{eq:untangling}; the second equality is by the pull-through equations for the UPT and its dual~\eqref{eq:cupcapmodsdualpnt}; and the third equality is by evaluation of $\dim(H_e)=d$.
The second equation of~\eqref{eq:entinvertibility} is proven similarly. 

The naturality equations~\eqref{eq:starcohomsnat} follow immediately from naturality of the UPTs $\alpha$ and $\alpha^*$~\eqref{eq:pntmonnat}.
\end{proof}
\begin{corollary}\label{cor:entequivs}
Let $f: [A_1,m_1,u_1] \to{} [A_2,m_2,u_2]$ be a channel between standard separable Frobenius algebras in $\Rep(G)$, and let 
\begin{align*}
F(f): F([A_1,m_1,u_1]) &\to F([A_2,m_2,u_2])\\
F'(f): F'([A_1,m_1,u_1]) &\to F'([A_2,m_2,u_2])
\end{align*}
be the corresponding channels in $\Chan(G)$ and $\Chan(G')$ respectively. Then: 
\begin{itemize}
\item $(u_{[A_1,m_1,u_1]},v_{[A_2,m_2,u_2]},H_e)$ is an entanglement-assisted channel coding scheme for $F(f)$ from $F'(f)$.
\item $(v_{[A_1,m_1,u_1]},u_{[A_2,m_2,u_2]},H_e)$ is an entanglement-assisted channel coding scheme for $F'(f)$ from $F(f)$.
\end{itemize}
That is, the following equations hold:
\begin{calign}\label{eq:starcohomentsymms}
\includegraphics[scale=0.9,valign=c]{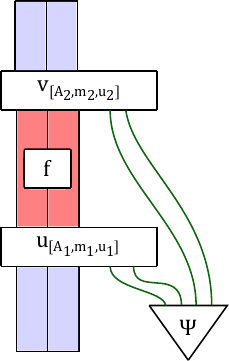}
~~=~~
\includegraphics[scale=0.9,valign=c]{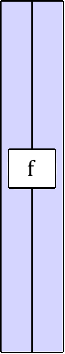}
~~~&&~~~
\includegraphics[scale=0.9,valign=c]{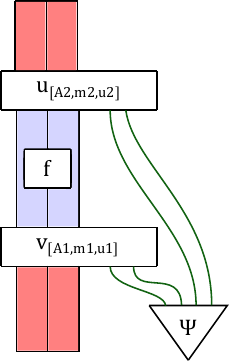}
~~=~~
\includegraphics[scale=0.9,valign=c]{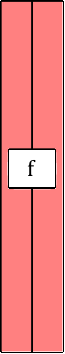}
\end{calign}
\end{corollary}
\begin{proof}
This follows straightforwardly from Theorem~\ref{thm:starcohoms}. For the first equation of~\eqref{eq:starcohomentsymms}:
\begin{calign}\nonumber
\includegraphics[scale=0.9,valign=c]{Figures/svg/entanglementsymms/starcohomentsymmeqs11.pdf}
~~=~~
\includegraphics[scale=0.9,valign=c]{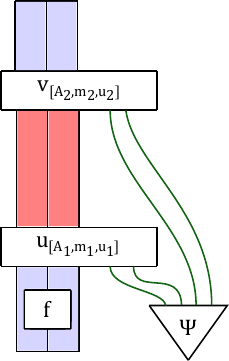}
~~=~~
\includegraphics[scale=0.9,valign=c]{Figures/svg/entanglementsymms/starcohomentsymmeqs12.pdf}
\end{calign}
Here the first equality is by naturality~\eqref{eq:starcohomsnat} and the second equality is by entanglement-invertibility~\eqref{eq:entinvertibility}. The second equation of~\eqref{eq:starcohomentsymms} is proven similarly.
\end{proof}
\noindent
Taken together, these entanglement-symmetries give rise to an equivalence of categories.
\begin{definition}\label{def:equivalence}
We define a map from objects and morphisms of $\CP(G)$ (resp. $\Chan(G)$) to objects and morphisms of $\CP(G')$ (resp. $\Chan(G')$) as follows. 

For each f.d. $G$-$C^*$-algebra $([A,m,u],\rho)$, choose a covariant unitary $*$-isomorphism $\iota_{([A,m,u],\rho)}: ([A,m,u],\rho) \to F([\tilde{A},\tilde{m},\tilde{u}])$ for some $[\tilde{A},\tilde{m},\tilde{u}]$ in $\Rep(G)$. (Such a $*$-isomorphism exists by Theorem~\ref{thm:equivalence}.)

The mapping is defined as follows:
\begin{itemize}
\item \emph{On objects}.
$
([A,m,u],\rho) \mapsto F'([\tilde{A},\tilde{m},\tilde{u}]).
$
\item \emph{On morphisms.} We map any covariant CP map (resp. channel) $f: ([A_1,m_1,u_1],\rho_1) \to ([A_2,m_2,u_2],\rho_2)$ to the following morphism $F'([\tilde{A}_1,\tilde{m}_1,\tilde{u}_1]) \to F'([\tilde{A}_2,\tilde{m}_2,\tilde{u}_2])$:
\begin{calign}\nonumber
\includegraphics[scale=1,valign=c]{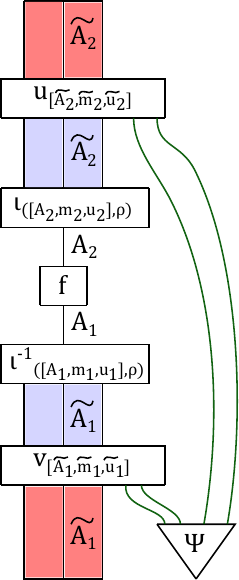}
\end{calign}
\end{itemize}
\end{definition}
\begin{theorem}\label{thm:starcohomequiv}
The map of Definition~\ref{def:equivalence} is an equivalence of categories $\CP(G) \cong \CP(G')$ (resp. $\Chan(G) \cong \Chan(G')$).
\end{theorem}
\begin{proof}
This follows straightforwardly from Theorem~\ref{thm:equivalence}; indeed, the map is precisely the composition of a weak inverse of the equivalence $\CP(\Rep(G)) \to \CP(G)$ (resp. $\Chan(\Rep(G)) \to \Chan(G)$) with the equivalence $\CP(\Rep(G)) \to \CP(G')$ (resp. $\Chan(\Rep(G)) \to \Chan(G')$).
\end{proof}

\section{Example and application}\label{sec:ex}
We now present a brief example and application of our construction in the case where the $G$-action corresponds to a finite group grading. 

\subsection{Example: finite group grading}

Because in a rigid $C^*$-tensor category all objects are  finite direct sums of simple objects, in order to specify such a category it is sufficient to specify only the isomorphism classes of simple objects and their fusion rules.
\begin{definition}
Let $G$ be a finite group. We define a rigid $C^*$-tensor category $\Hilb(G)$ as follows. 
\begin{itemize}
\item Simple objects $[g],[h],\dots$: indexed by group elements $g,h,\dots \in [g]$.
\item Tensor product: $[g] \otimes [h] := [gh]$, i.e. given by the product in the group.
\item Tensor unit: $[e]$, where $e \in G$ is the identity.
\item Dual objects: the dual of $[g]$ is $[g^{-1}]$, where $g^{-1} \in G$ is the inverse element in the group. 
\item Canonical fibre functor: maps each simple object to the 1-dimensional Hilbert space $\mathbb{C}$.
\end{itemize}
\end{definition}
\noindent 
The compact quantum group associated to $\Hilb(G)$ is the Hopf $*$-algebra $C(G)$ of complex functions on the group $G$. We now recall the classification of fibre functors on $\Hilb(G)$.  First we define an equivalence relation on pairs $(L,[\psi])$ of a subgroup $L<G$ and a 2-cohomology class $[\psi] \in H^2(L, U(1))$, as follows:
\begin{align}\nonumber
(L_1,[\psi_1]) \sim (L_2,[\psi_2]) \qquad \Leftrightarrow \qquad L_2 = g L_1 g^{-1} &\textrm{~and~} \psi_1 \textrm{~is cohomologous to~} 
\\\label{eq:ctequivrel}
&\psi_2^g(x,y):= \psi_2(g x g^{-1}, gyg^{-1}) \textrm{~for some~}g \in G
\end{align}
An equivalence class of rank-one module categories over $\Hilb(G)$ is determined by a 2-cohomology class $[\psi] \in H^2(G,U(1))$, up to the equivalence relation~\eqref{eq:ctequivrel}~\cite[Ex. 7.4.10]{Etingof2016}. To be explicit, a rank-one module category has a single simple object $X$; the action of the simple objects of $\Hilb[G]$ on $X$ simply preserves $X$, i.e. $[g] \overline{\otimes} X = X$ for all $g \in G$, while the associativity constraint is  specified by the 2-cocycle $\psi$:
\begin{align*}
m_{g,h}:= \psi(g,h) \id_{X}: (g \otimes h) \overline{\otimes} X \to g \overline{\otimes} (h \overline{\otimes} X) 
\end{align*}
A rank-one module category is precisely a fibre functor, since an action of a tensor category $\mathcal{T}$ on a category $\mathcal{M}$ is the same thing as a tensor functor $\mathcal{T} \to \End(M)$~\cite[Prop. 7.1.3]{Etingof2016}. We thereby obtain the following characterisation of the fibre functors $F_{\psi}$ on $\Hilb(G)$.
\begin{lemma}
Let $\psi \in Z^2(G,U(1))$ be some 2-cocycle. Then $F_{\psi} \cong F \circ E_{\psi}$, where $F: \Hilb(G) \to \Hilb$ is the canonical fibre functor and $E_{\psi}$ is an autoequivalence of $\Hilb(G)$, specifically the identity functor equipped with the following generally nontrivial multiplicator:
\begin{align*}
m_{[g_1],[g_2]} = \overline{\psi(g_1,g_2)} ~\id_{[g_1g_2]}: E_{\psi}([g_1]) \otimes E_{\psi}([g_2]) \to E_{\psi}([g_1 g_2])
\end{align*}
\end{lemma}
\noindent
We will now compute the autoequivalence of $\CP(C(G))$ associated to the fibre functor $F_{\psi}$. We first recall the classification of $\F$s in $\Hilb(G)$. Every $\F$ is a direct sum of indecomposable $\F$s~\cite[Lem. 3.11]{Verdon2021}, so we need only classify the indecomposable $\F$s. 
\begin{definition}
For any subgroup $L<G$ and 2-cocycle $\phi \in Z^2(L, U(1))$, let $A(L,\phi)$ be the following $\F$ in $\Hilb(G)$:
\begin{itemize}
\item Underlying object: $\bigoplus_{g \in L} [g]$.
\item Multiplication: specified on the factors by $[g] \otimes [h] \overset{\phi(g,h) \id_{[gh]}}{\to} [gh] $.
\item Unit: $\begin{pmatrix} \id_{[e]} & 0 & \dots & 0 \end{pmatrix}^T : [e] \to \bigoplus_{g \in L} [g]$.
\end{itemize}
We call $A(L,\phi)$ the \emph{$\phi$-twisted group algebra for $L$}.
\end{definition}
\noindent
By~\cite[Prop. 3.1]{Ostrik2003}, indecomposable $\F$s in $\Hilb(G)$ are all of the form $V \otimes A(H,\phi) \otimes V^*$ for some pair $(L,\phi)$ and some object $V$ of $\Hilb(G)$, where the multiplication and unit of the $\F$ are defined using the multiplication and unit of $A(H,\phi)$ together with the cup and cap of the duality on $V$:
\begin{align*}
\includegraphics[scale=1]{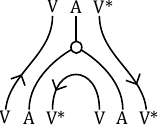}
\end{align*}
\begin{proposition}\label{prop:autoequivindecomp}
Let $\psi \in Z^2(G,U(1))$ be a 2-cocycle. Up to unitary natural isomorphism, the autoequivalence of $\CP(C(G))$ associated to the fibre functor $F_{\psi}$ is defined as follows on indecomposable $\F$s:
\begin{itemize}
\item On objects: the indecomposable $\F$ $V \otimes A(L,\phi) \otimes V^*$ is taken to the indecomposable $\F$ $V \otimes A(L,\overline{\psi} \phi) \otimes V^*$.
\item On morphisms: a CP morphism $f: V_1 \otimes A(L_1,\phi_1) \otimes V_1^* \to V_2 \otimes A(L_2,\phi_2) \otimes V_2^*$ is taken to $v_2 \circ F_{\psi}(f) \circ v_1^{\dagger}$, where \begin{align*}
v_1: F_{\psi}(V_1 \otimes A(L_1,\phi_1) \otimes V_1^*) &\to V_1 \otimes A(L_1,\overline{\psi}\phi_1) \otimes V_1^* \\
v_2: F_{\psi}(V_2 \otimes A(L_2,\phi_2) \otimes V_2^*) &\to V_2 \otimes A(L_2,\overline{\psi}\phi_2) \otimes V_2^*
\end{align*}
are canonical unitary isomorphisms of $\F$s defined in~\eqref{eq:unitaryFisos}.
\end{itemize}
\end{proposition}
\begin{proof}
We observe that the multiplication of the algebra $F_{\psi}(V \otimes A(L,\phi) \otimes V^*)$ is defined as follows:
\begin{align*}
\includegraphics[scale=1,valign=c]{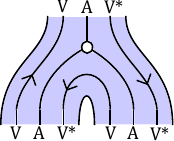}
~~=~~
\includegraphics[scale=1,valign=c]{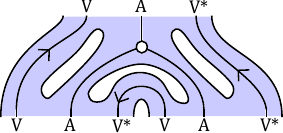}
\end{align*}
It is therefore unitarily isomorphic to the $\F$ on $F_{\psi}(V) \otimes F_{\psi}(A(L,\phi)) \otimes F_{\psi}(V^*)$ with the following multiplication and unit:
\begin{align}\label{eq:equivFdef}
\includegraphics[scale=1,valign=c]{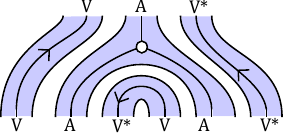}
&&
\includegraphics[scale=1,valign=c]{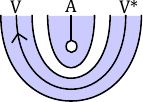}
\end{align}
Here we have used unitarity of the fibre functor $F_{\psi}$ (Definition~\ref{def:fibre}).
Now we observe that the following cup and cap appearing in~\eqref{eq:equivFdef} are standard, so they are are related to the original cup and cap on $V$ by unitary isomorphism~\cite[Prop. 2.2.15]{Neshveyev2013}:
\begin{align*}
\includegraphics[scale=1,valign=c]{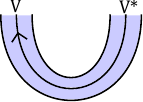}
~~=~~
\includegraphics[scale=1,valign=c]{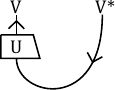}
&&
\includegraphics[scale=1,valign=c]{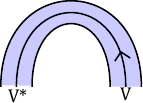}
~~=~~
\includegraphics[scale=1,valign=c]{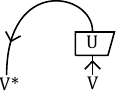}
\end{align*}
It follows that there is a unitary isomorphism of $\F$s $F_{\psi}(V \otimes A(L,\phi) \otimes V^*) \cong V \otimes A(L,\overline{\psi} \phi) \otimes V^*$, as claimed.  The effect of the equivalence on a CP morphism is to conjugate it by these unitary isomorphisms:
\begin{align}\label{eq:unitaryFisos}
\includegraphics[scale=1,valign=c]{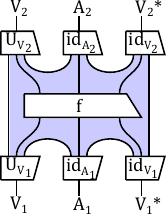}
\end{align}
\end{proof}
\begin{example}[Abelian groups of central type: from classical to quantum channels]\label{ex:gradedchans}
Let $G$ be an finite abelian group. Consider the (un)twisted group algebra $A(G,1)$, where $1$ is the trivial 2-cocycle on $G$. 

Under the canonical fibre functor on $\Hilb(G)$, this can be considered as a $G$-graded commutative f.d. $C^*$-algebra, equipped with an inner product such that the direct sum decomposition into one-dimensional homogeneously graded subspaces is orthogonal. There are two orthonormal bases of interest for this algebra. The first is the graded basis $\{\ket{g}~|~g \in G\}$, which has one element for each component of the grading. The second is the factor basis; that is,  the basis realising the isomorphism of f.d. $C^*$-algebras $A(G,1) \cong \oplus_{i=1}^{|G|}\mathbb{C}$. In the factor basis, a channel $f$ on $A(G,1)$ is precisely an $|G| \times |G|$ stochastic matrix $(f_{i,j})_{ij}$, i.e. a matrix with entries $f_{i,j} \in [0,1]$ satisfying $\sum_{i} f_{i,j} = 1$ for all $j \in \{1,\dots,|G|\}$.

Explicitly, the factor basis is the basis of characters $\{\ket{\chi}~|~ \chi \in G^*\}$, where 
\begin{align}\label{eq:factorbasis}
\ket{\chi}:= \frac{1}{\sqrt{|G|}} \sum_{g \in G} \chi(g) \ket{g}.
\end{align} There is a representation  of the group of characters $G^*$ on the algebra $A(G,1)$ given in the factor basis by $\chi_1 \cdot \ket{\chi_2} = \ket{\chi_1 \chi_2}$, for all $\chi_1,\chi_2 \in G^*$. A channel $f$ on $A(G,1)$ preserving the grading is precisely a stochastic matrix which is invariant under this action, i.e. satisfying:
\begin{align}\label{eq:gradedchan}
f_{\chi_i,\chi_j} = f_{\chi \chi_i,\chi \chi_j} ~~~~~\forall~ \chi,\chi_i,\chi_j \in G^*
\end{align}
Let us suppose that $G$ is an abelian group of \emph{central type}; that is, it possesses a 2-cocycle $\psi$ such that $A(G,\psi)$ is a graded matrix algebra (i.e. a $G$-graded algebra $B(H)$, where $\dim(H) = \sqrt{|G|}$). (For example, let $G = \mathbb{Z}_2 \times \mathbb{Z}_2$, and let $\psi$ be the 2-cocycle corresponding to multiplication of the Pauli matrices.) Then the autoequivalence associated to the functor $F_{\overline{\psi}}$ takes the graded f.d. $C^*$-algebra $A(G,1)$ to the graded f.d. $C^*$-algebra $A(G,\psi) \cong B(H)$, which has the same graded basis, but where the elements now multiply as $\ket{g} \cdot \ket{h} = \psi(g,h) \ket{gh}$. These elements of $B(H)$ form a \emph{nice unitary error basis} for $H$~\cite{Knill1996}, which we will write as $\{u_g~|~g \in G\}$.

Let us compute the image of a classical channel $f$ preserving the grading under this autoequivalence. Observe that the map $\rho(g,h):= \psi(g,h)\overline{\psi(h,g)}: G \times G \to U(1)$ is a bicharacter; in particular, the map 
\begin{align*}
\rho_R: G &\to G^*\\
g &\mapsto \rho(g,-)
\end{align*}
is an isomorphism of groups. 
Now the change of basis matrix $W$ from the factor basis to the graded basis in $A(L,1)$ is the Fourier transform:
$$
W_{g,\rho_R(h)}:= \frac{1}{\sqrt{|G|}}\rho_R(h)(g) = \frac{1}{\sqrt{|G|}} \rho(h,g)
$$
In the graded basis, the grading-preserving channel $f$ has the following entries:
\begin{align*}
f_{g,h} = (W f W^{\dagger})_{g,h} &= \frac{1}{|G|} \sum_{k,l \in G} \rho(k,g) f_{\rho_R(k),\rho_R(l)} \overline{\rho(l,h)} 
\\
&= \frac{1}{|G|} \sum_{k,l \in G} \rho(k,g) f_{1,\rho_R(k^{-1}l)} \rho(h,l) 
\\
&= \frac{1}{|G|} \sum_{k,x \in G} \rho(k,g) f_{1,\rho_R(x)} \rho(h,kx) 
\\
&= \frac{1}{|G|} \sum_{k,x \in G} \overline{\rho(g,k)}\rho(h,k) f_{1,\rho_R(x)} \rho(h,x) 
\\
&= \delta_{g,h} \sum_{x \in G} f_{1,\rho_R(x)} \rho(g,x) 
\end{align*}
Here for the second equality we used~\eqref{eq:gradedchan} and the definition of the bicharacter $\rho$;  for the fourth equality we used that $\rho$ is a bicharacter; and for the fifth equality we used orthogonality of characters for $\rho(g,-)$ and $\rho(h,-)$. By Proposition~\ref{prop:autoequivindecomp}, the autoequivalence corresponding to $F_{\psi}$ takes $f$ to a channel $f'$ with identical coefficients in the graded basis (since $V$ is trivial, so the unitary isomorphism in that proposition is trivial). 
It follows that the quantum channel $f'$ obtained by the entanglement-symmetry multiplies the unitary error basis elements $u_g$ by a scalar factor:
\begin{align}\label{eq:fdashquantum}
f'(u_g) = \sum_{x \in G} f_{1,\rho_R(x)} \rho(g,x) u_g
\end{align}
\end{example}
\begin{example}[From classical-to-quantum channels to quantum-to-classical channels]
In this example we will show that entanglement-symmetries can relate quantum-to-classical channels (channels from a matrix algebra to a commutative f.d. $C^*$-algebra) with classical-to-quantum channels (channels from a commutative f.d. $C^*$-algebra to a matrix algebra). From the perspective of entanglement-assisted channel coding, this is to say that a channel which transmits classical information based on measurement of a quantum state is an equivalent communication resource, in the entanglement-assisted setting, to a channel which transmits quantum information based on a classical input.

We will be somewhat more concrete in this example and work with the group $G:=\mathbb{Z}_2 \times \mathbb{Z}_2$. This is an abelian group of central type; the nondegenerate 2-cocycle $\psi$ arises from the projective representation of $\mathbb{Z}_2 \times \mathbb{Z}_2$ as Pauli matrices:
\begin{align}\nonumber
(0,0) \mapsto u_{00}:= \begin{pmatrix} 1 & 0 \\ 0 & 1 \end{pmatrix} && (1,0) \mapsto u_{10}:= \begin{pmatrix} 0 & 1 \\ 1 & 0 \end{pmatrix} 
\\\label{eq:paulicocycle}
(0,1) \mapsto u_{01}:=  \begin{pmatrix} 0 & -i \\ i & 0 \end{pmatrix} && (1,1) \mapsto u_{11}:=  \begin{pmatrix} 1 & 0 \\ 0 & -1 \end{pmatrix}
\end{align}
We observe that $\psi^2$ is cohomologous to the trivial 2-cocycle $1$. In terms of algebras, this corresponds to the isomorphism of $\F$s
\begin{align}\nonumber
A(G,\psi^2) &\to A(G,1) \\\label{eq:omegaisom}
\ket{g} &\mapsto \omega(g) \ket{g}
\end{align}
where $\omega((0,0)) = 1$ and $\omega(g) = -1$ for all $g \neq (0,0) \in G$.

By Proposition~\ref{prop:autoequivindecomp}, a covariant channel $f: A(G,1) \to A(G,\psi)$ will therefore be taken to a covariant channel $f': A(G,\psi) \to A(G,1)$ by the entanglement symmetry corresponding to the fibre functor $F_{\overline{\psi}}$. 

We will first describe the grading-preserving channels $f: A(G,1) \to A(G,\psi)$. Such a channel is defined by a choice of quantum state --- i.e a trace-1 positive element in $A(G,\psi)$ --- for each element of the factor basis of $A(G,1)$.  If the channel is to be grading-preserving, it must satisfy $\ket{g} \mapsto a_g u_g$ for every $g \in G$, where $\{a_g\}_{g \in G}$ are some complex numbers. Given the description~\eqref{eq:factorbasis} of the factor basis, we see that the channel is defined by
$$
\ket{\chi} \mapsto \frac{1}{2} \sum_{g \in G} a_g \chi(g) u_g.
$$
For these elements to be trace-1, we need $a_{00} = 1$. By explicit computation we find that the condition for positivity of these elements is
\begin{align*}
a_{10},a_{01},a_{11} \in \mathbb{R} && 1 - \sqrt{(a_{10})^2+(a_{01})^2+(a_{11})^2} \geq 0
\end{align*}
The coefficients describing such channels therefore form a unit ball in $\mathbb{R}^3$.

We will now calculate the quantum-to-classical channel $f'$ to which such a channel $f$ will be mapped by the entanglement-symmetry. By Proposition~\ref{prop:autoequivindecomp}, in the graded basis it will have exactly the same coefficients; that is, it will (after composition with the isomorphism~\eqref{eq:omegaisom}) map the element $u_{g} \in A(G, \psi)$ to the element $\omega(g) a_g \ket{g} \in A(G, 1)$. The channel is described by positive operators (POVM elements) $M_{\chi} \in A(G,\psi)$ for $\chi \in G^*$, such that 
$$f'(x) = \sum_{\chi \in G^*}2 \Tr(M_{\chi} x) \ket{\chi}.$$ 
(Here the scalar factor $2$ arises from our consideration of the separable trace.)

Since $\{u_g\}_{g \in G}$ is a basis for $A(G,\psi)$, we can write 
$$M_{\chi} = \sum_{g \in G} c_{\chi,g} u_g^{\dagger}$$
for some complex scalars $\{c_{\chi,g}\}_{\chi \in G^*, g \in G}$. By~\eqref{eq:factorbasis} we have $$\ket{g} = \frac{1}{2} \sum_{h\in G} \rho(g,h) \ket{\rho_R(h)},$$ and therefore $u_g \mapsto \omega(g)a_g\ket{g}$ implies
$$
\frac{\omega(g) a_g}{2} \rho(g,h) = 2\Tr(M_{\rho_R(h)} u_g) = 4 c_{\rho_R(h),g}~~~~\Rightarrow~~~~ c_{\rho_R(h),g} = \frac{\omega(g) a_g}{8}\rho(g,h)
$$
Here for the second equality we used orthogonality of $\{u_g\}_{g \in G}$ under the trace inner product. To check that this is indeed a POVM we observe that 
$$
2 \sum_{h \in G} M_{\rho_R(h)} = 2\sum_{h \in G} \frac{\omega(g)a_g}{8} \rho(g,h) u_g^{\dagger} 
= \omega((0,0))a_{00} \mathbbm{1} = \mathbbm{1}.
$$
Here for the second equality we used orthogonality of the characters $\rho(g,-)$. In summary, we see that the quantum-to-classical channel $f'$ has POVM elements
$$
M_{\chi} = \frac{1}{8}\sum_{g \in G} \omega(g) a_g \overline{\chi(g)} u_g^{\dagger}.
$$
\end{example}
\begin{remark}
We will briefly explain how to obtain the UPTs realising these entanglement-symmetries. Since $C(G)$ is a finite compact quantum group, every fibre functor is accessible from the canonical fibre functor by a unitary pseudonatural transformation~\cite[Cor. 3.17]{Verdon2020}. For a 2-cocycle $\psi \in Z^2(G,U(1))$, a UPT $F \to F_{\psi}$ can be straightforwardly calculated using~\cite[Thm. 4.12]{Verdon2020}. Explicitly, let $\pi: G \to B(H)$ be an irreducible projective representation of $G$ with cocycle $\psi$.  Let $[g]$ be a simple object of $\Hilb[G]$; we observe that $F(V) = F_{\psi}(V) = \mathbb{C}$. We then define a UPT $(\alpha,H): F \to F_{\psi}$ with components $\alpha_{[g]}: H \to H$ by $\alpha_{[g]}:= \pi(g)$. The $*$-cohomomorphism $F(A) \otimes B(H) \to F_{\psi}(A)$ can then be computed for any f.d. $C^*$-algebra $A$ using~\eqref{eq:udef}.
\end{remark}
\ignore{
For the following proposition, which is an immediate application of Theorem~\ref{}, we fix some notation. By Theorem~\ref{}, for any system $\tilde{A}$ in $\CP(C(G))$ the algebras $A = \tilde{F}(\tilde{A})$ and $\tilde{F}_{\psi}(\tilde{A})$ are identical as graded Hilbert spaces; the only difference is the twist in the multiplication. Let $A_g$ be the homogeneous subspace of $A$ for the group element $g \in G$. There are embeddings $i_g: A_g \to\tilde{F}(\tilde{A})$ and $i_g': A_g \to \tilde{F}_{\psi}(\tilde{A})$. Let $I_g := i_g' \circ i_g^{\dagger}: \tilde{F}(\tilde{A}) \to \tilde{F}_{\psi}(\tilde{A})$.
\begin{proposition}
Let $\tilde{A}$ be a system in $\CP(C(G))$. Let $\tilde{F}$ be the canonical fibre functor on $\CP(C(G))$, and let $\tilde{F}_{\psi}$ be the fibre functor associated to some 2-cocycle $\psi \in Z^2(G,U(1))$.

Then the entanglement-invertible channel $(M_{\tilde{A}},H): \tilde{F}(\tilde{A}) \to \tilde{F}_{\psi}(\tilde{A})$ and its inverse $(N_{\tilde{A}},H): \tilde{F}_{\psi}(\tilde{A}) \to \tilde{F}(\tilde{A})$ are defined as follows:
\begin{align*}
M_{\tilde{A}}(\rho \otimes \sigma) = \frac{1}{\sqrt{\dim(H)}} \sum_{g \in G} \Tr[\pi(g)\sigma] I_g (\rho)  &&
N_{\tilde{A}}(\rho \otimes \sigma) = \frac{1}{\sqrt{\dim(H)}} \sum_{g \in G} \Tr[\pi(g)^* \sigma] I_g^{\dagger} (\rho)
\end{align*}
\end{proposition}
}
\begin{remark}
These are not the only entanglement-symmetries associated to a finite group. Indeed, rather than the category $\Hilb$ of $G$-graded Hilbert spaces one can consider the category $\Rep(G)$ directly; now the systems possess a $G$-action rather than a $G$-grading. In this case the interesting fibre functors arise when $G$ is a group of central type. Again, all fibre functors are accessible from the canonical fibre functor by a unitary pseudonatural transformation. We will not go into details here as the description of the fibre functors is more complicated.
\end{remark}

\subsection{Application: Entanglement-assisted capacities of quantum channels}\label{sec:application}

Quantum channels have several distinct capacities, such as the classical capacity $C$ and the quantum capacity $Q$. Following~\cite{Bennett1999}, we are interested in $C_E$, the \emph{entanglement-assisted classical capacity} of a quantum channel. This is a quantum channel's capacity for transmitting classical information with the help of unlimited prior entanglement between sender and receiver. (Note that $C_E$ determines the analogous entanglement-assisted quantum capacity $Q_E$, since $Q_E = C_E/2$ by teleportation and dense coding.) 

It is well-known that, for classical channels, $C_E = C$; entanglement cannot increase the classical capacity. However, for quantum channels this no longer holds. The following proposition allows $C_E$ to be computed precisely for certain quantum channels.
\begin{proposition}\label{prop:computecap}
Let $G$ be a compact quantum group, and let $f: A \to B$ be a $G$-covariant channel. Let $F'$ be any fibre functor on $\Rep(G)$ which is accessible by a unitary pseudonatural transformation, and let $E: \Chan(G) \to \Chan(G')$ be the corresponding equivalence. Then $C_E[E(f)] = C_E[f]$. 
\end{proposition}
\begin{proof}
By Corollary~\ref{cor:entequivs} we have entanglement-assisted channel coding schemes interchanging the two channels. 
\end{proof}
\begin{example}
One straightforward way to apply Proposition~\ref{prop:computecap} is to start with some covariant classical channels whose capacity is known. Let $X,Y$ be finite sets. A \emph{weakly symmetric} classical channel $f:X \to Y$ is one whose stochastic matrix $p(y|x)$ satisfies the following conditions:
\begin{itemize}
\item All rows are permutations of each other. 
\item The channel is unital (i.e. it preserves the uniform distribution).  
\end{itemize}
The capacity of a weakly symmetric classical channel is known to be
\begin{align}\label{eq:weaksymmcap}
C = \log_2(|Y|) - H(\textrm{any row of stochastic matrix})
\end{align}
where $H$ is the Shannon entropy.

We will use Example~\ref{ex:gradedchans} to compute the entanglement-assisted capacities of some quantum channels. Let $G$ be an abelian group of central type and let $f$ be a grading preserving channel on $A(G,1)$; that is, a $|G| \times |G|$ matrix whose rows and whose columns are indexed by elements of $G^*$, satisfying the following conditions:
\begin{itemize}
\item $f_{\chi_i,\chi_j} \in [0,1]~~~\forall \chi_i,\chi_j \in G^*$.
\item $\sum_{\chi \in L^*} f_{\chi,\chi_j} = 1~~\forall \chi_j \in G^*$ (stochastic).
\item $f_{\chi \chi_i,\chi \chi_j} = f_{\chi_i,\chi_j}~~\forall \chi_i,\chi_j \in G^*$ ~~(grading-preserving).
\end{itemize}
Then $f$ is a weakly symmetric classical channel on the commutative f.d. $C^*$-algebra $A(G,1)$; it is weakly symmetric because it is $G$-graded, and unital because it is stochastic:
\begin{align*}
\frac{1}{|G|}\sum_{\chi \in G^*} f_{\chi_i,\chi} = \frac{1}{|G|}\sum_{\chi \in G^*} f_{\chi^{-1} \chi_i,1} = \frac{1}{|G|}
\end{align*}
Its entanglement-assisted classical capacity is therefore given by the formula~\eqref{eq:weaksymmcap}.

Let $\psi$ be a nondegenerate 2-cocycle on $G$, and let $\{u_g~|~g \in G\}$ be the corresponding nice unitary error basis of $B(H)$, where $\dim(H) = \sqrt{|G|}$. Then the channel $f'$ on $B(H)$ defined by~\eqref{eq:fdashquantum} has the same entanglement-assisted classical capacity as the classical channel $f$.

Let us see which channels this gives us the entanglement-assisted capacity for when $G = \mathbb{Z}_2 \times \mathbb{Z}_2$ with the Pauli 2-cocycle~\eqref{eq:paulicocycle}. A $G$-graded stochastic matrix $f$ is defined by its first row $\{f_{1,\chi}\}_{\chi \in G^*}$; these coefficients are free in $[0,1]$ provided that $\sum_{\chi \in G^*} f_{1,\chi} = 1$, i.e. they are a probability distribution over the elements of $G^*$. The capacity of the channel is then defined by 
\begin{align}\label{eq:chancapac}
C = 2-H((f_{1,\chi})_{\chi \in G^*}).
\end{align}
By~\eqref{eq:fdashquantum}, the corresponding quantum channel on $B(H)$ is defined by the following map on Pauli matrices:
$$u_{g} \mapsto \sum_{\chi \in G^*} f_{1,\chi}\overline{\chi(g)} u_g$$
The entanglement-assisted capacity of such channels is given precisely by the formula~\eqref{eq:chancapac}.
\end{example}

\bibliographystyle{alphaurl}
\bibliography{bibliography}
\appendix
\section{Appendix: Faithful traces on f.d. $G$-$C^*$-algebras} \label{sec:app}

In physics, it is common to assume that the faithful positive linear functional on a f.d. $C^*$-algebra is tracial. This is precisely to demand that the corresponding Frobenius algebra should be symmetric.
In the case where there is no $G$-action, one loses no generality since there is always a tracial state on the algebra. In fact, a normalisation can be chosen for the trace that will pick it out uniquely; here we follow~\cite{Vicary2011} and demand that the Frobenius algebra should be separable~\eqref{eq:frobspecial}.
\begin{proposition}\label{prop:specialtrace}
Every f.d. $C^*$-algebra $A$ admits a unique \emph{separable} trace making the corresponding symmmetric Frobenius algebra $[A,m,u]$ separable. Moreover, with this separable trace $u^{\dagger} u = \dim(A)$.
\end{proposition}
\begin{proof}
Every f.d. $C^*$-algebra $A$ admits an orthogonal decomposition as a multimatrix algebra, i.e. $A \cong \oplus_{i = 1}^n M_{n_i}(\mathbb{C})$. Let $\{p_i\}$, $\{\iota_i\}$ be the corresponding orthogonal projections and injections onto and from the factors, and define $\Tr_i: M_{n_i}(\mathbb{C}) \overset{\iota_i}{\hookrightarrow} A  \to[\Tr] \mathbb{C}$. Then we have $\Tr(x) = \Tr(\sum_i \iota_i p_i x)= \sum_i \Tr_i(p_i x)$, so $\Tr$ is determined by the $\Tr_i: M_{n_i}(\mathbb{C}) \to \mathbb{C}$ on each factor. \ignore{Moreover the trace inner product satisfies $\Tr(a^{\dagger}b) = \sum_i \Tr( a^\dagger \iota_i p_i b) = \sum_i \Tr( (p_i a)^{\dagger} p_i b)$, implying the decomposition $A \cong \oplus_{i = 1}^n M_{n_i}(\mathbb{C})$ is orthogonal.} Since the multiplication $m: A \otimes A \to A$ has the form $m = \sum_i \iota_i\circ m_i \circ (p_i \otimes p_i)$, the adjoint comultiplication has the form $m^{\dagger} = \sum_i (\iota_i \otimes \iota_i) \circ m_i^{\dagger} \circ p_i$, and we see that $m \circ m^{\dagger} = \sum_i \iota_i \circ m_i \circ m_i^{\dagger} \circ p_i$. It follows that $\Tr$ will be separable on $A$ iff the $\Tr_i$ are separable on the factors. 

Recall that the trace on a matrix algebra is unique up to scaling. Consider the matrix algebra $M_{n_i}(\mathbb{C})$. Under the inner product resulting from the canonical trace $\overline{\Tr_i}$ it is straightforward to show  that $m_i \circ m_i^{\dagger} = n_i \mathbbm{1}$. Scaling the trace by $n$ scales $m^{\dagger}$ by $n^{-1}$ (this is shown in~\cite[Lem. 4.4]{Vicary2011}). Therefore the unique trace on $M_{n_i}(\mathbb{C})$ yielding a separable Frobenius algebra is $n_i \overline{\Tr_i}$. It follows that the the unique choice of trace on $A$ yielding separability is:
$$
\Tr(x) = \sum_{i} n_i \overline{\Tr_i}(p_i x)
$$

Finally, for the last statement we observe that $u^{\dagger} u = \Tr(1^*1) = \Tr(1) = \sum_i n_i\overline{\Tr_i}(1_i) = \sum_i n_i^2 = \dim(A)$.
\end{proof}
\noindent
In the case of a f.d. $G$-$C^*$-algebra for a compact quantum group $G$ the situation becomes more interesting. Although every f.d. $G$-$C^*$-algebra admits an invariant functional (in particular, we can pick the canonical invariant functional from Lemma~\ref{lem:uniquespecialfunctional}), not every f.d. $G$-$C^*$-algebra will admit a $G$-invariant faithful trace.

We first observe that any Frobenius algebra $[A,m,u]$ in $\Rep(G)$ corresponding to a given f.d. $G$-$C^*$-algebra $A$ is $*$-isomorphic to any other, so in particular their quantum dimension is identical. It therefore makes sense to talk about the quantum dimension $\dim_q(A)$ of a f.d. $G$-$C^*$-algebra (Definition~\ref{def:quantumdimension}). Likewise, we can consider the ordinary dimension $\dim_c(A)$ of a f.d. $G$-$C^*$-algebra; this is the dimension of the object $F(A)$ in $\Hilb$, where $[A,m,u]$ is a Frobenius algebra in $\Rep(G)$ corresponding to the f.d. $G$-$C^*$-algebra $A$ with some choice of invariant functional and $F: \Rep(G) \to \Hilb$ is the canonical fibre functor. 
 
We remark that the separable trace of Proposition~\ref{prop:specialtrace} is precisely the canonical $G$-invariant functional on a f.d. $G$-$C^*$-algebra for trivial $G$.
\begin{lemma}\label{lem:invspectrace}
Let $\phi$ be a $G$-invariant faithful trace on a f.d. $G$-$C^*$-algebra $A$. Then there exists a $G$-invariant faithful trace on $A$ making the corresponding Frobenius algebra in $\Rep(G)$ separable. 
\end{lemma}
\begin{proof}
We use~\cite[Lem. 2.8]{Neshveyev2018}, which states that any Frobenius algebra in $\Rep(G)$ is unitarily $*$-isomorphic to a direct sum of simple Frobenius algebras. We can therefore reduce to the case of a simple Frobenius algebra with $G$-invariant trace. In this case we already have $m \circ m^{\dagger} = \lambda \id$ for some $\lambda>0$, so using a scalar $*$-isomorphism ($m \mapsto \frac{1}{\sqrt{\lambda}} m$ and $u \mapsto \sqrt{\lambda} u$) we obtain a $G$-invariant faithful trace on $A$ making the corresponding Frobenius algebra separable. 
\end{proof}
\noindent
\begin{theorem}\label{thm:caninvtraceisspec}
A f.d. $G$-$C^*$-algebra $A$ admits a $G$-invariant faithful trace iff $\dim_q(A) = \dim(A)$. In this case the canonical $G$-invariant functional is tracial and is precisely the separable trace of Proposition~\ref{prop:specialtrace}.
\end{theorem}
\begin{proof}
We first show that if $\dim_q(A)=\dim(A)$ then the canonical $G$-invariant functional is tracial. Consider the separable Frobenius algebra $[A,m,u]$ in $\Rep(G)$ corresponding to $A$ with its canonical $G$-invariant functional. If we forget the $G$-action we obtain a separable Frobenius algebra in $\Hilb$ satisfying $u^{\dagger}u = \dim_q(A) = \dim(A)$. But then this is precisely a finite-dimensional $C^*$-algebra equipped with its canonical functional, which was shown to be tracial in Proposition~\ref{prop:specialtrace}.

We now show that if $\dim_q(A) > \dim(A)$ (by~\cite[Cor. 2.2.20]{Neshveyev2013}, it can only be greater), then the f.d. $G$-$C^*$-algebra $A$ admits no $G$-invariant faithful trace. We will prove this by contradiction. Suppose that $A$ admits a $G$-invariant faithful trace. Then by Lemma~\ref{lem:invspectrace} it admits a $G$-invariant faithful trace making the corresponding Frobenius algebra separable; we may therefore consider a separable Frobenius algebra $[A,m,u]$ in $\Rep(G)$ such that the associated positive linear functional is tracial.

We observe that a separable Frobenius algebra $[A,m,u]$ in $\Rep(G)$ satisfies $\dim_q(A) \leq u^{\dagger}u$, with equality iff the Frobenius algebra is standard. For this we use~\cite[Thm. 2.2.19]{Neshveyev2013}. Indeed, we note that, in the language of that theorem, $R=\overline{R}= m^{\dagger}u$ is a solution of the conjugate equations for $A$. We then have $\dim_q(A) \leq u^{\dagger}mm^{\dagger}u=u^{\dagger}u$ with equality iff $m^{\dagger}u$ is a standard solution of the conjugate equations up to a scalar factor; the scalar factor must be trivial since $R=\overline{R}$.

Forgetting the $G$-action, we obtain a separable Frobenius algebra in $\Hilb$ such that the faithful positive linear functional is tracial. But in Proposition~\ref{prop:specialtrace} it was shown that the only such positive linear functional is the canonical one, satisfying $u u^{\dagger} = \dim(A)$. But since we know that $u^{\dagger} u \geq \dim_q(A) > \dim(A)$, we arrive at a contradiction. There can therefore be no $G$-invariant faithful trace on $A$.  
\end{proof}
\end{document}